\documentclass[
aps,
prb,
reprint,twocolumn,floats,floatfix,
nofootinbib]{revtex4-2}

\usepackage{lipsum}
\usepackage{graphicx}
\usepackage{dcolumn}
\usepackage{bm}
\usepackage{amsmath}
\usepackage{amsfonts}
\usepackage{amssymb}
\usepackage{amstext}
\usepackage[english]{babel}
\usepackage{helvet}
\usepackage{microtype}
\usepackage[pdftex,hidelinks]{hyperref}
\hypersetup{
    colorlinks=true,
    linkcolor=blue,
    filecolor=blue,
    urlcolor=blue,
    citecolor = blue,
}
\usepackage[export]{adjustbox}
\usepackage[caption=false]{subfig}
\usepackage{xcolor}
\usepackage{amsthm}
\usepackage{mathtools}
\usepackage{bbold}
\usepackage{float}
\usepackage{enumerate}

\graphicspath{{./images/}}

\newcommand{\ket}[1]{|#1\rangle}
\newcommand{\bra}[1]{\langle#1|}
\newcommand{\braket}[2]{\langle#1|#2\rangle}
\newcommand{\average}[1]{\langle #1 \rangle}
\DeclareMathOperator{\sign}{sign}
\DeclareMathOperator{\diag}{diag}
\DeclarePairedDelimiter\floor{\lfloor}{\rfloor}

\newtheorem{theorem}{Theorem}

\theoremstyle{remark}
\newtheorem*{remark}{Remark}

\begin{document}


\title{Lattice and $\mathcal{PT}$ symmetries in tensor-network renormalization group:\\ Case study of a hard-square lattice gas model}

\author{Xinliang Lyu}
\email[]{xlyu@ihes.fr}
\affiliation{
    Institut des Hautes \'Etudes Scientifiques,
    91440 Bures-sur-Yvette, France
}


\date{July 12, 2026}

\begin{abstract}
The tensor-network renormalization group (TNRG) is an accurate numerical real-space renormalization group method for studying phase transitions in both quantum and classical systems.
Continuous phase transitions, as an important class of phase transitions, are usually accompanied by spontaneous breaking of various symmetries.
However, the understanding of symmetries in the TNRG is well-established mainly for global on-site symmetries like $U(1)$ and $SU(2)$.
In this paper, we demonstrate how to incorporate lattice symmetries (including reflection and rotation) and the $\mathcal{PT}$ symmetry in the TNRG in two dimensions (2D) through a case study of the hard-square lattice gas with nearest-neighbor exclusion.
This model is chosen because it is well-understood and has two continuous phase transitions whose spontaneously-broken symmetries are lattice and $\mathcal{PT}$ symmetries.
Specifically, we write down proper definitions of these symmetries in a coarse-grained tensor network and propose a TNRG scheme that incorporates these symmetries.
We demonstrate the validity of the proposed method by estimating the critical parameters and the scaling dimensions of the two phase transitions of the model.
The technical development in this paper has made the 2D TNRG a more well-rounded numerical method.
\end{abstract}

\maketitle


\section{Introduction}
Tensor-network renormalization group (TNRG) is a versatile modern formulation of the real-space renormalization group (RG) for both classical and quantum models~\cite{Levin:Nave:2007,Xie:2012:hotrg,Evenbly:2017algo,Okunishi:2022review}.
As a powerful numerical method, it can be used for evaluating the free energy of a classical statistical system and the expectation value of observables of a quantum system.
Understood in the Wilsonian RG framework, the TNRG generates RG flows in the space of tensors that reveal novel phases of matter~\cite{Gu:Wen:2009}.
At criticality, the coarse-grained tensors that flow near the critical fixed point can be used to extract scaling dimensions and operator product expansion (OPE) coefficients of the universality class~\cite{Gu:Wen:2009,Evenbly:2016local, Lyu:Xu:Kawashima:2021,Guo:Wei:2024}.
Apart from being a numerical tool, the TNRG is also becoming a rigorous method for studying the stability of various RG fixed points~\cite{Kennedy:2022,Kennedy:2023,Ebel:2024-3D,Ebel:2025-meets}.
In this paper, we focus on the TNRG as a numerical real-space RG method applying to the partition function of classical statistical systems.

In numerical calculations, it is preferable to incorporate the symmetries of a physical system in its tensor-network representation and to preserve them in TNRG manipulations~\cite{Singh:2010,Singh:2011,Singh:2012,Lyu:Kawashima:2025-3D,Lyu:Kawashima:2025:reflsym}.
This can reduce the computational time since there are fewer degrees of freedom due to the constraints of the symmetries on the tensors.
Furthermore, since the space of the RG flows is constrained after incorporating the symmetries, all the local operators that do not respect the symmetries are eliminated in numerical RG; this can improve the quality of the RG flows and facilitate the numerical search of a critical fixed-point tensor~\cite{Lyu:Kawashima:2025-3D,Ebel:2025:newton}.

The study of the symmetries in TNRG as a numerical real-space RG method is unavoidably influenced by the benchmark model used to gauge the behavior and efficiency of the method.
In two dimensions (2D), the benchmark model has always been the 2D Ising model since it is exactly solved.
The symmetry that is spontaneously broken in the 2D Ising transition is the global on-site spin-flip $\mathbb{Z}_2$ symmetry.
This might be one of the reasons why such global on-site symmetries like $U(1)$ and $SU(2)$~\cite{Singh:2010,Singh:2011,Singh:2012} were quickly established in tensor network decomposition\footnote{
$\mathbb{Z}_N$ symmetry is closely related to $U(1)$ symmetry.
} a few years after the birth of the first TNRG scheme, tensor network renormalization (TRG)~\cite{Levin:Nave:2007}.
However, the first detailed study of the lattice symmetries in the TNRG had not appeared until almost two decades later~\cite{Lyu:Kawashima:2025:reflsym}, where only the lattice-reflection symmetry was discussed.
We suspect that this lack of motivation for understanding other symmetries in the TNRG is because the 2D Ising model is used exclusively as the benchmark model, and a mild break of lattice symmetries in numerical calculation does not cause any problem for the RG flows and the accuracy of the numerical estimation.
Using other models as benchmarks in the development of the TNRG can stimulate the study of other symmetries, making the TNRG a more general and versatile numerical RG method.

Hard-core lattice gas models are systems of particles moving on a lattice with a hard-core interaction~\cite{Dhar:2023}.
These models are important for understanding the phase transitions from a fluid-like phase to a solid-like ordered phase, which is one of the most familiar phase transitions seen in everyday life.
When the range of the hard-core interaction increases, the phase transitions of these models become richer and more intricate.
For example, a naive symmetry argument predicts that the model with the next-nearest-neighbor (2NN) exclusion on a square lattice should have a critical point belonging to the 4-state Potts model~\cite{Dhar:2018talk};
yet, the critical exponents found in the Monte Carlo simulations are close to those of the 2D Ising model~\cite{Fernandes:Heitor:2007}.
This confusion was settled later when people noticed a sliding instability and columnar order in the high-density phase of the model~\cite{Ramola:2012-hizexpansion}.
This transition is located on the critical line of the Ashkin-Teller criticality connecting the Ising and the 4-state Potts transitions and is closer to the Ising transition~\cite{Ramola:2012-phd,Ramola:2015-columnar}.

The main numerical tools for studying these hard-core lattice gas models are Monte Carlo simulations and the transfer matrix method~\cite{Dhar:2023}.
The previous TNRG studies of such models~\cite{Akimenko:2019,Akimenko:2023} focus on calculating the density of the gas by evaluating the partition function using the TRG, without incorporating the symmetries of the models.
In this paper, we will make the first step of unleashing the full power of the TNRG as an efficient real-space RG method for studying hard-core lattice gas.
The emphasis is on incorporating an advanced TNRG technique called entanglement filtering (EF)~\cite{Gu:Wen:2009,Evenbly:2015tnr,Evenbly:2017algo,loop-TNR:2017,Bal:2017,Evenbly:2018:fet,Hauru:2018,Harada:2018,Lee:Kawashima:2020ring,Homma:2024} in the simple TRG while preserving important symmetries of the model that are previously not well understood in the TNRG.
The EF is essential for applying the TNRG to a system at its criticality; it reduces the RG truncation errors and makes the tensor stable near a critical fixed point~\cite{Evenbly:2015tnr,Hauru:2018}.

Specifically, we use the hard-square lattice gas with nearest-neighbor exclusion (1NN) to guide the understanding of lattice-reflection, lattice-rotation, and $\mathcal{PT}$ symmetries~\cite{Meisinger:2013,Bender:Boettcher:1998,Bender:Hook:2024} in the TNRG.
The 1NN hard-square lattice gas is the simplest hard-core lattice gas model on a square lattice.
As will be explained in Sec.~\ref{sec:model}, this model has two phase transitions---one is physical and the other nonphysical.
The physical transition is accompanied by a spontaneous breaking of lattice symmetries, while for the nonphysical one, the $\mathcal{PT}$ symmetry is broken.
Although this model is not exactly solved, the most accurate estimates of the critical parameters of its two phase transitions have more than ten significant digits~\cite{Todo:Suzuki:1996,Guo:Blote:2002,Todo:1999}.
Moreover, the two transitions are known to belong to the 2D Ising class and the Yang-Lee edge singularity, whose universal data, like scaling dimensions, are exactly known.
Therefore, it is a good benchmark model.

The remaining part of the paper is organized as follows.
In Sec.~\ref{sec:model}, we introduce the 1NN hard-square lattice gas model and construct a tensor-network representation of its partition function.
In order to make the lattice-reflection, lattice-rotation and $\mathcal{PT}$ symmetries of the model manifest, the resultant tensor network has copies of a nontrivial bond matrix on the bonds of the tensor network.
Moreover, to demonstrate the importance of incorporating these symmetries in the TNRG, we conduct numerical experiments to study the stability of the RG flows near the spontaneous symmetry breaking (SSB) fixed points using TNRG maps with different symmetry properties.
In Sec.~\ref{sec:TNRG}, we write down a proper definition of these symmetries in a coarse-grained tensor network and propose an EF-enhanced TNRG map that both preserves and imposes these symmetries.
The proposed scheme is based on the TRG and its key ingredient for incorporating the lattice symmetries is to understand how to perform the singular value decomposition (SVD) splitting of a tensor in a way that does not hide these lattice symmetries.
Enhanced by an EF called loop optimization~\cite{loop-TNR:2017,Homma:2024}, our scheme is a generalization of the symmetric loop-TNR in Ref.~\cite{loop-TNR:2017}.
We demonstrate the efficiency of the proposed method by estimating the critical parameter and scaling dimensions in the two phase transitions of the model in Sec.~\ref{sec:numres}.

\section{The model and its tensor-network representation\label{sec:model}}
We present a tensor-network representation for the partition function of the hard-square model with nearest-neighbor exclusion (1NN).
The phase transition of this model occurs when the activity, which controls the density of particles, changes.
There are two known transition points: one at positive activity and the other at negative.
We emphasize the relevant symmetries associated with these two transition points and how these symmetries manifest themselves in the tensor-network representation of the model.
Numerical experiments are conducted to demonstrate that incorporating these symmetries in an RG map should improve the stability of the fixed point corresponding to the spontaneous-symmetry-breaking (SSB) phase.

\subsection{The known phase transitions of the model\label{sec:model:transitions}}
The hard-square lattice gas model describes particles populating the sites of a square lattice in two dimensions (2D).
Each site can only hold $0$ or $1$ particle.
The only interaction of this system is the nearest-neighbor exclusion (1NN): no two particles are adjacent to each other.
Due to this infinite interaction, temperature dependence is absent.
Hence, the only thermodynamic state variable of the system is the density of the particles: $\rho = \average{n} / N$, where $\average{n}$ is the average number of particles and $N$ the number of lattice sites.
Due to the 1NN interaction, the square lattice is naturally divided into two sublattices like a chessboard shown in Fig.~\ref{fig:model-chessboard}(a).

\begin{figure}[tb]
\subfloat[]{
    \includegraphics[width=0.45\columnwidth,
    valign=c]{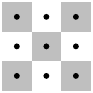}
}
~
\subfloat[]{
    \includegraphics[width=0.45\columnwidth,
    valign=c]{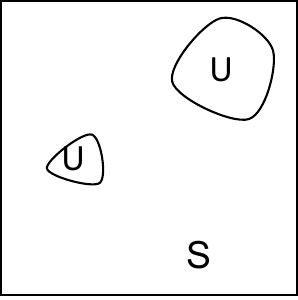}
}
\caption{\label{fig:model-chessboard}
        (a) Black dots represent sites of the square lattice, and they are divided into two sublattices like a chessboard.
        If a particle is in the central shaded site, it excludes other particles from its four neighbor unshaded sites.
        (b) A large-length-scale picture of the 1NN hard-square model at large activity $z$.
        ``S'' denotes the area where particles occupy the shaded sublattice, while ``U'' is the area where particles occupy the unshaded sublattice.
    }
\end{figure}

The system is described by the following grand canonical partition function~\cite{Baxter:1999}:
\begin{align}
    \label{eq:model:Zdef}
    \mathcal{Z}_N(z) = \sum_{n=0}^N g(n,N) z^n,
\end{align}
where $g(n,N)$ is the number of ways of putting $n$ particles on a square lattice with $N$ sites and the activity $z$ is the probability weight of having one particle in the system.
In the thermodynamic limit, the free energy of the system can be defined to be
\begin{align}
    \label{eq:model:fdef}
    f(z) = \lim_{N \to \infty}
    \frac{\log \mathcal{Z}_N(z)}{N}.
\end{align}
In order for the partition function $\mathcal{Z}_N(z)$ and free energy $f(z)$ to have a clear physical picture, the activity $z$ should be non-negative.
However, if one understands phase transition in terms of distribution of zeros of $\mathcal{Z}_N(z)$ on the complex $z$ plane as per the Yang-Lee theory of phase transition~\cite{Yang:Lee:1952a,Yang:Lee:1952b}, or as singular behavior of the free energy $f(z)$, it also makes sense to consider negative activity $z<0$.

For positive activity $z$, this model undergoes a second-order phase transition from a disordered fluid-like phase to an ordered solid where one sublattice has more particles than the other.
It is known that this transition belongs to the 2D Ising universality class, which can be understood using the following argument~\cite{Dhar:2018talk}.
At small $z$ near zero, the density of the particle is low, so the system is a nearly-free gas;
two sublattices have the same density of particles in this phase.
When $z \to +\infty$, the density goes to the maximal value $\rho=1/2$, all particles occupy either one or the other sublattice in Fig.~\ref{fig:model-chessboard}(a).
In this high-density phase, the density of the particles $\rho$ will be different for two sublattices.
Suppose all particles occupy the shaded sublattice at $z = + \infty$.
When $z$ decreases to a finite but large value, fluctuations will create regions of the lattice where particles populate the unshaded sublattice.
In a large length scale, the physical picture of the system looks like a sea of particles at the shaded sublattice, with small islands where particles populate the unshaded sublattice (see Fig.~\ref{fig:model-chessboard}(b)).
This is exactly the same physical picture as the Ising model when the temperature increases from zero temperature, where the two competing ground states are spin-up and spin-down areas.
This argument has been checked in a Monte Carlo simulation using data collapse~\cite{Fernandes:Heitor:2007}.
The critical activity $z_c^+ \approx 3.796255$, as well as the two relevant scaling dimensions, has been estimated to a very high accuracy (with 14 significant digits) by a finite-size analysis of transfer matrix calculations~\cite{Todo:Suzuki:1996,Guo:Blote:2002}.

It is known that the activity expansion of the free energy $f(z)$ has its radius of convergence dominated by a singularity on the negative-$z$ axis due to the alternating sign of the coefficients in the series~\cite{Groeneveld:1962}; 
the location of this singularity $z_c^- \approx -0.119389$ has also been accurately estimated with 10 significant digits using numerical diagonalization of transfer matrices~\cite{Todo:1999} and a series expansion from corner transfer matrix renormalization group~\cite{Chan:2012,Jensen:2012}.
Although a negative value of activity is nonphysical, the singular behavior of the free energy $f(z) \sim (z-z_c^-)^{\phi}$ near $z_c^-$ is universal for particles with repulsive interactions, with the exponent $\phi$ depending only on the dimensionality $d$ of the system~\cite{Poland:1984,Baram:Luban:1987}.
Therefore, this non-physical singularity is known as the \emph{repulsive-core singularity}.
At $d=2$, the value is known exactly to be $\phi(d=2)=5/6$ from Baxter's exact solution of the hard-hexagon lattice gas~\cite{Baxter:1980}.

This repulsive-core singularity can be understood in the framework of Yang-Lee theory of phase transition~\cite{Yang:Lee:1952a,Yang:Lee:1952b}, where the singular behavior of $f(z)$ is determined by the distribution of zeros of the partition function $\mathcal{Z}_N(z)$ on the complex $z$ plane in the thermodynamical limit.
The zeros of the partition function are expected to form lines on the complex $z$ plane, and the density of these zeros diverges like $(s-s_e)^{\sigma}$ at the edges of the lines, where $s$ is some parametrization along a line and $s_e$ is one edge of the line~\cite{Kortman:Griffiths:1971}.
This singularity of zeros, known as \emph{Yang-Lee edge singularity}, is also universal and only depends on system dimensionality~$d$~\cite{Fisher:1978}.
The scaling of this singularity is described by a scalar field theory with $\varphi^3$ interaction and a purely imaginary coupling.
Later, Fisher and Lai~\cite{Lai:Fisher:1995}, by comparing the values of $\phi$ and $\sigma+1$ in different $d$, realized that the repulsive-core singularity is identical to Yang-Lee edge singularity, and the two exponents are related to each other by $\phi = \sigma + 1$.
Recent numerical study of the partition function zeros of the 1NN hard-square model has clearly shown the universal hard-core line on the negative-$z$ axis~\cite{Assis:2014}.
The Yang-Lee edge singularity corresponds to the simplest non-unitary Conformal Field Theory (CFT), the $M(2,5)$ minimal model~\cite{Cardy:1985}.
This minimal model has only one relevant field with scaling dimension $x_{\varphi}=-2/5$, which is related to the critical exponent according to $\sigma = x_{\varphi} / (2 - x_{\varphi})$~\cite{Fisher:1978}.
Cardy has recently written a review about the history of this topic~\cite{Cardy:2024}.

\subsection{The relevant symmetries of the two phase transitions\label{sec:model:symPhase}}
\begin{figure}[tb]
\centering
\includegraphics[scale=1.0]{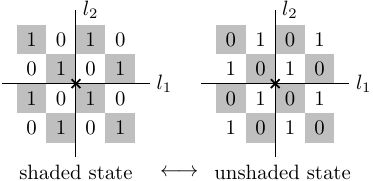}
\caption{\label{fig:model-zpground}
    SSB of lattice symmetries for positive-$z$ transition.
    Number $0, 1$ on the site means number of particle.
    When the activity $z=+\infty$, all particles occupy either the shaded (left) or unshaded (right) sites.
    These two configurations can be mapped to each other through any of the following symmetry transformations: 
    1) translation by one lattice constant in any direction,
    2) $90^\circ$ lattice rotation around the central cross point, and
    3) lattice reflection along either the $l_1$ or the $l_2$ axis.
    }
\end{figure}


Phase transitions between ordered and disordered phases are often accompanied by spontaneous symmetry breaking (SSB)~\cite{Kardar:2007:book}.
In the ordered phase, typical configurations of the system, as an ensemble, break some symmetry of the Hamiltonian (energy functional) of the system.
More precisely, in classical statistical mechanics, SSB can be understood as the degeneracy of the eigenstates (ground states) of the transfer matrix with the largest eigenvalue under the thermodynamical limit.
In the disordered phase, the ground state is unique. 
In ordered phase, the ground states are degenerate and they are mapped into each other by some \emph{global} symmetry operations of the transfer matrix\footnote{The basis states of the ground-state subspace are not connected by a finite number of local operations.
}.
The system will then stay in one of the ground states according to its initial or boundary condition and thus breaks the global symmetry of the transfer matrix.
The classical example of SSB is the second-order phase transition of the nearest-neighbor interaction Ising model, where the low-temperature phase breaks the global spin-flip $\mathbb{Z}_2$ symmetry of its Hamiltonian.

To study the RG flows of a model using TNRG, it is important to incorporate into the RG transformation the symmetry relevant to SSB of a phase transition.
Without incorporating the symmetry, the fixed point corresponding to the ordered phase becomes unstable under the RG transformation.
The numerical error due to artifacts of the RG scheme or the machine precision would break the degeneracy of the ordered-phase fixed point.
For example, the low-temperature fixed point of the 2D Ising model becomes unstable and eventually flows away to a fixed-point tensor with no degeneracy if the spin-flip $\mathbb{Z}_2$ symmetry is not incorporated in the TNRG~\cite{Gu:Wen:2009}.

For the positive-$z$ transition, the high-density ordered phase spontaneously breaks the symmetry between two sublattices.
For the hard-square model, there is no Hamiltonian.
A symmetry of the system means a set of symmetry transformations that bring an allowed configuration to another one with the same statistical weight.
In this sense, the system has all the lattice symmetries, including translation, rotation, and reflection.
The ordered phase can be understood as the breaking of these lattice symmetries, as is shown in Fig.~\ref{fig:model-zpground}.
These lattice symmetries can be made more explicit after introducing the tensor-network representation of this model.

The repulsive-core singularity at negative activity is related to SSB of the $\mathcal{PT}$ symmetry of the model~\cite{Bender:Boettcher:1998,Bender:Hook:2024,Dorey:2009}.
The $\mathcal{PT}$-symmetric models in statistical mechanics are a class of models whose partition function is real but can be negative~\cite{Meisinger:2013}.
Due to the lack of physical interpretation of negative probability, this symmetry is understood mathematically as the following symmetry property of the transfer matrix $T$ of the model:
\begin{align}
    \label{eq:model:TmatPTsym}
    \mathcal{P} T \mathcal{P} = T^*,
\end{align}
where the parity operator $\mathcal{P}$ is unitary and $\mathcal{P}^2=1$.
The above symmetry of the transfer matrix can also be written as $[\mathcal{PT}, T]=0$, with the time-reversal operator $\mathcal{T}$ implemented as complex conjugation.
Moreover, the parity and time reversal operators commute with each other $[\mathcal{P}, \mathcal{T}]=0$.
For the hard-square model with a negative activity $z<0$, its transfer matrix clearly can be constructed as a real matrix.
In this real representation, the parity operator is trivial: $\mathcal{P}=1$.
Due to this symmetry, the eigenvalues of the transfer matrix $T$ can either be real or appear as complex-conjugate pairs.

\begin{figure}[tb]
\centering
\includegraphics[width=1.0\columnwidth]{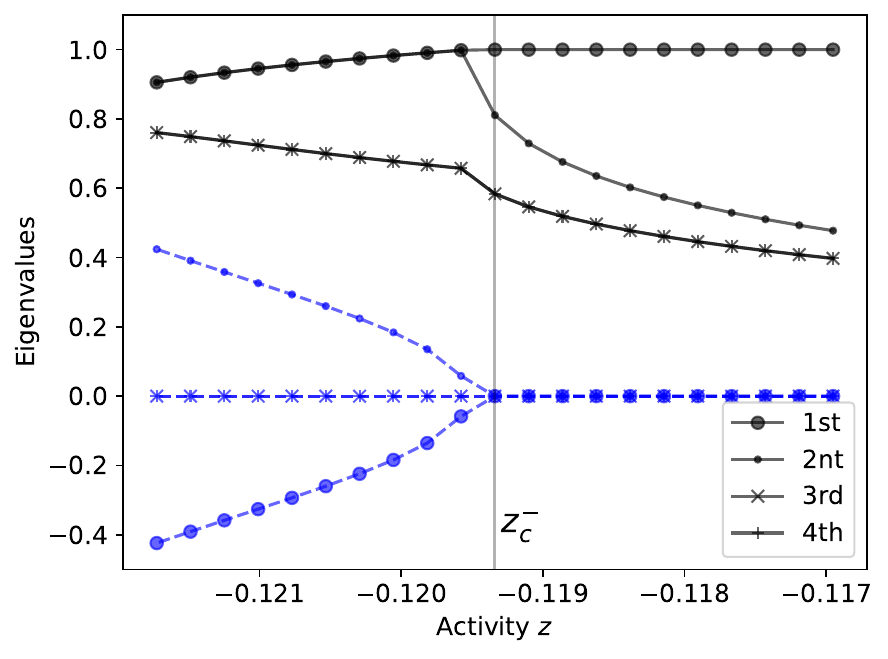}
\caption{\label{fig:model-TMspectz}
    The four largest eigenvalues of the transfer matrix of the 1NN hard-square model with size $L=12$ for activity near the repulsive-core singularity $z_c^{-}$.
    The largest eigenvalues are normalized to have unit norm.
    Black solid lines are real parts, and blue dashed lines are imaginary parts.
    When the activity $z$ decreases and goes over $z_c^{-}$, the first two eigenvalues develop imaginary parts and become a complex-conjugate pair, which is SSB of the $\mathcal{PT}$ symmetry.
        Notice that the third and fourth eigenvalues are degenerate, and thus they overlap each other in the plot.
    }
\end{figure}

At $z=0$, the only non-vanishing eigenvalue of the transfer matrix $T$ is 1.
When $z$ decreases to some negative value, the second-largest eigenvalue will cross with the largest one, after which their imaginary parts start to grow and they become a complex-conjugate pair (see Fig.~\ref{fig:model-TMspectz}).
The value of $z=z_c^-$ at which the first two largest eigenvalues cross is expected to be the Yang-Lee edge singularity~\cite{Todo:1999,Dhar:2018talk} in the thermodynamical limit\footnote{If the eigenvalue with the largest absolute value is unique, then the partition function cannot be zero.
Only after the degeneracy occurs is it possible for the partition function to be zero.
}.
For $z<z_c^{-}$, the two degenerate states with the largest eigenvalues are mapped to each other by the global $\mathcal{PT}$ symmetry operation.
In this sense, the $\mathcal{PT}$ symmetry is spontaneously broken.

\subsection{Tensor-network representation of the model and its symmetries}
There are various approaches for constructing the initial tensor network of the partition functions of hard-core lattice gas models~\cite{Akimenko:2019,Akimenko:2023}.
Since the emphasis of the current study is the RG flows of the model, which is sensitive to the symmetries incorporated into the RG map, we will make the relevant symmetries of the two transitions explicit in the initial tensor network.
As we shall see, the lattice symmetries that are relevant for the positive-$z$ transition correspond to lattice symmetries of the initial tensor network.
The $\mathcal{PT}$ can be most easily incorporated by demanding that the tensors should be real-valued.
The transfer matrix $T$, which can be constructed from the real-valued tensors, is also real, and the parity operator in Eq.~\eqref{eq:model:TmatPTsym} becomes trivial $\mathcal{P}=1$.

The basic convention in a tensor-network diagram is that if a tensor leg is shared by two tensors, this leg is summed over (or contracted).
Therefore, a full contraction of a tensor network can represent the partition function of a system.
We claim that the initial tensor network of the partition function in Eq.~\eqref{eq:model:Zdef} is
\begin{align}
    \label{eq:model:initialTN}
    &\mathcal{Z}_N(z) = \mathcal{Z}(A_0(z), \sigma_0;L_x, L_y) \nonumber\\
    & \equiv
    \includegraphics[scale=1.0, valign=c]{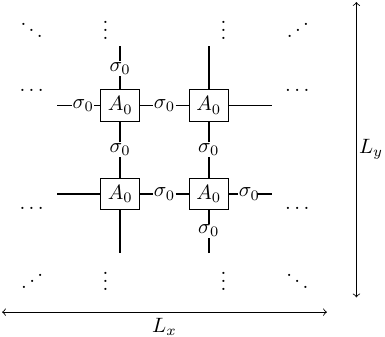}.
\end{align}
    In this tensor-network representation, the square lattice formed by the tensor $A_0$ is tilted by $45^\circ$ from the original square lattice formed by the sites of the lattice gas in Fig.~\ref{fig:model-chessboard}(a).
    The sites of the lattice gas are located at the position of the bond matrix $\sigma_0$.
    In Eq.~\eqref{eq:model:initialTN}, the matrix $\sigma_0$ is placed on the bonds extending outwards in such a way that when this basic $2 \times 2$ block is repeated in the space, there is one bond matrix on each bond.

The $4$-leg tensor $A_0$ in Eq.~\eqref{eq:model:initialTN} is
\begin{subequations}
\label{eqsub:model:A0}
\begin{align}
    \label{eq:model:A0}
    \includegraphics[scale=1.0, valign=c]{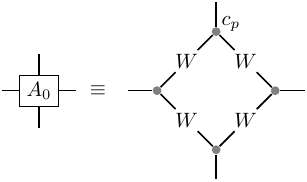}\quad,
\end{align}
where the 1NN matrix $W$ is
\begin{align}
    \label{eq:model:W1NN}
    W=
    \begin{pmatrix}
    W_{00}=1 & W_{01}=1  \\
    W_{10}=1 & W_{11}=0
    \end{pmatrix},
\end{align}
and the 3-leg COPY-dot~\cite{Biamonte:2011,Hauru:2018,Akimenko:2019,Akimenko:2023} $c_p$ carries the absolute value of the activity $z$, with only two non-vanishing components:
\begin{align}
    \label{eq:model:copydot}
    \includegraphics[scale=1.0, valign=c]{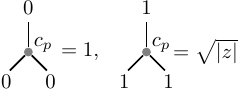}\quad.
\end{align}
\end{subequations}
The zero entry of the 1NN matrix in Eq.~\eqref{eq:model:W1NN},  $W_{11}=0$, means that no two particles can be adjacent to each other, while the COPY-dot encodes the particle number on the lattice site where it sits.
Later, we will show that it is more natural to put the activity on the COPY-dot $c_p$ than on the 1NN matrix $W$.
In Eq.~\eqref{eq:model:initialTN}, the bond matrix $\sigma_0$ is diagonal and carries the sign of the activity $z$:
\begin{align}
    \label{eq:model:sigma0}
    (\sigma_0)_{00}=1, (\sigma_0)_{11} = \sign{z}.
\end{align}
Therefore, when $z>0$, the bond matrix becomes trivial: $\sigma_0=\mathbb{1}$;
when $z<0$, the bond matrix is the diagonal Pauli matrix $\sigma_0 = \sigma^z = \diag(1, -1)$.
For $z=0$, the value of $(\sigma_0)_{11}$ does not affect the partition function in Eq.~\eqref{eq:model:initialTN} and we can set $(\sigma_0)_{11}=1$, the same as $z>0$ case.
In this tensor-network representation of the partition function, the number of lattice sites $N$ in the square lattice in Fig.~\ref{fig:model-chessboard}(a) is twice that of the tensor $A_0$ in Eq.~\eqref{eq:model:initialTN}: $N = 2L_x L_y$.

One can derive the tensor network representation in Eq.~\eqref{eq:model:initialTN} according to the following procedure.
On the lattice of the hard-square model, put copies of the 1NN matrix $W$ in Eq.~\eqref{eq:model:W1NN} on the bond and copies of a $4$-leg COPY-dot $c$ on the lattice site;
the full contraction of the resultant tensor network represents the partition function in Eq.~\eqref{eq:model:Zdef}:
\begin{align}
    \label{eq:model:wc2z}
    \mathcal{Z}_N(z) =
    \includegraphics[scale=1.0, valign=c]{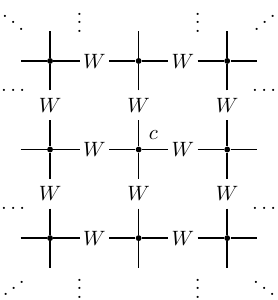}\quad,
\end{align}
where the 4-leg COPY-dot $c$ has only two non-vanishing components:
\begin{align}
    \label{eq:model:4legcdot}
    \includegraphics[scale=1.0, valign=c]{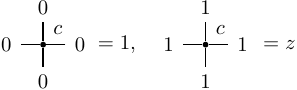}.
\end{align}
The physical meaning of the tensor network in Eq.~\eqref{eq:model:wc2z} is straightforward.
The 1NN matrix $W$ encodes the nearest-neighbor exclusion.
The 4-leg COPY-dot in Eq.~\eqref{eq:model:4legcdot} says that if no particle sits on the site, the weight is 1, while if one particle occupies the site, the weight is activity $z$.
At this point, one can see that it is more natural to put the activity on the $4$-leg COPY-dot $c$ than the $W$ matrix, since copies of $c$ sit on the lattice sites where particles also sit.
The tensor network in Eq.~\eqref{eq:model:wc2z} can be put into the form in Eq.~\eqref{eq:model:initialTN} by splitting the 4-leg COPY-dot diagonally in two ways\footnote{This splitting is closely related to the SVD-splitting in the TRG, where the $\sigma_0$ matrix is absorbed into one $c_p$ tensor.
}:
\begin{align}
    \label{eq:model:c2cp}
    \includegraphics[scale=1.0, valign=c]{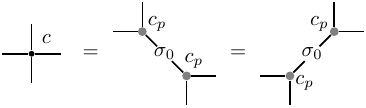},
\end{align}
where $c_p$ and $\sigma_0$ are given in Eqs.~\eqref{eq:model:copydot} and~\eqref{eq:model:sigma0}.
After applying the first splitting to one sublattice and the second splitting to the other and combining $W$ and $c_p$ to become the 4-leg tensor $A_0$ in Eq.~\eqref{eq:model:A0}, one arrives at the tensor network in Eq.~\eqref{eq:model:initialTN}, $45^\circ$ rotated.

Both the lattice symmetries and the $\mathcal{PT}$ symmetry of the hard-square model have a simple manifestation in the tensor network representation of its partition function in Eq.~\eqref{eq:model:initialTN}.
The $90^\circ$ lattice-rotation symmetry around the central cross point in Fig.~\ref{fig:model-zpground} is represented by the following symmetry of the initial tensor $A_0$ in Eq.~\eqref{eq:model:A0}:
\begin{align}
    \label{eq:model:A0rotsym}
    \includegraphics[scale=1.0, valign=c]{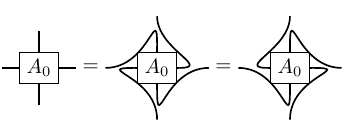}\quad.
\end{align}
This rotation symmetry of the initial tensor $A_0$ is a special case of a more general definition of the rotation symmetry of a coarse-grained tensor in Eq.~\eqref{eq:scheme:strongsym}, as will be discussed later in Sec.~\ref{sec:TNRG:defsym}.
The lattice-reflection symmetry along the $l_1$ and $l_2$ axes becomes the following diagonal reflection symmetry of the tensor $A_0$ in Eq.~\eqref{eq:model:A0}:
\begin{align}
    \label{eq:model:A0reflsym}
    \includegraphics[scale=1.0, valign=c]{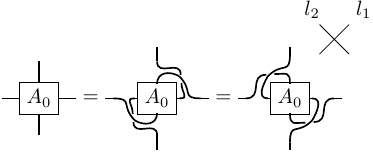}\quad.
\end{align}
This lattice-reflection symmetry along the $l_1$ and $l_2$ axes remains true even for a coarse-grained tensor, as will be shown later in Sec.~\ref{sec:TNRG:defsym}.
At infinite activity $z=+\infty$, only two configurations contribute to the partition function:
\begin{align}
    \label{eq:model:A0ground}
    \includegraphics[scale=1.0, valign=c]{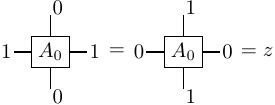}.
\end{align}
These two configurations of the tensor are mapped to each other by a $90^\circ$ rotation of the tensor or a reflection along either the $l_1$ or the $l_2$ axis.

The $\mathcal{PT}$ symmetry of the model is manifested in this tensor-network representation as the fact that the tensors $A_0,\sigma_0$ are real-valued, $A_0= A_0^*, \sigma_0 = \sigma_0^*$:
\begin{align}
    \label{eq:model:A0real}
    \includegraphics[scale=1.0, valign=c]{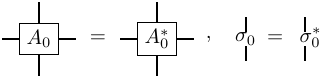}\quad,
\end{align}
where the complex conjugation is the time-reversal operator $\mathcal{T}$, while the parity operator $\mathcal{P}$ can be understood as the identity matrix acting on the tensor legs.
To see why Eq.~\eqref{eq:model:A0real} is a proper definition of $\mathcal{PT}$ symmetry, notice that the transfer matrix of the model $T$ can be built from the tensors $A_0,\sigma_0$ according to
\begin{align}
    \label{eq:model:A02TM}
    \includegraphics[scale=0.85, valign=c]{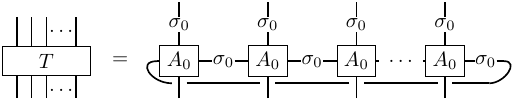}~.
\end{align}
Therefore, the transfer matrix $T$ is real-valued due to Eq.~\eqref{eq:model:A0real}.
It is straightforward to preserve and impose the $\mathcal{PT}$ symmetry of the tensors in TNRG by making sure that all tensors are real-valued in every tensor manipulation.

    We want to point out that the initial tensor-network representation in Eq.~\eqref{eq:model:initialTN} is different from the usual form of the tensor network in most existing TNRG schemes where the bond matrix is absent.
    The representation in Eq.~\eqref{eq:model:initialTN} is a generalized form where a nontrivial diagonal matrix on the bond can be taken into consideration.
    This kind of bond matrix appears naturally when one wants to exploit both the lattice-rotation and $\mathcal{PT}$ symmetries in the tensor-network representation.

\subsection{Stability of the spontaneous-symmetry-breaking phase under RG maps\label{sec:model:numdemo}}
In Sec.~\ref{sec:model:symPhase}, we point out the relevant symmetries related to the two phase transitions of the 1NN hard-square model.
The positive-$z$ transition is accompanied by SSB of lattice symmetries, as is illustrated in Fig.~\ref{fig:model-zpground}.
For the negative-$z$ transition (repulsive-core singularity), the $\mathcal{PT}$ symmetry is spontaneously broken (see Fig.~\ref{fig:model-TMspectz}).
Incorporating these symmetries in an RG transformation is important for studying the RG flows of the model, without which the RG fixed point corresponding to the ordered phase becomes unstable.
Before proposing a full-fledged TNRG scheme that incorporates all the relevant symmetries, we will conduct a series of numerical experiments to demonstrate the behavior of the RG flows generated by existing TNRG maps with distinct symmetry properties.
The numerical results here can be reproduced using the {\tt Python} codes published at Ref.~\cite{Lyu:algo:lattGas}.

To visualize the tensor RG flows, we define the following single-number characterization of the phase, known as a \emph{degeneracy index}~\cite{Gu:Wen:2009}:
\begin{align}
    \label{eq:model:degIndDef}
    X \equiv
    \frac{(\sum_k |\lambda_k|)^2}{\sum_k |\lambda_k|^2},
\end{align}
where $\{\lambda_k\}$ is the eigenvalue spectrum of the transfer matrix constructed from the coarse-grained 4-leg tensor $A$ and bond matrix $\sigma$ according to Eq.~\eqref{eq:model:A02TM} and Appendix~\ref{app:buildTM}.
In the thermodynamic limit, this number $X$ is the degeneracy of the eigenvalues of the transfer matrix with the largest absolute value.

\begin{figure*}[tb]
\subfloat[The RG flows generated by the HOTRG]{
    \includegraphics[scale=0.8,
    valign=c]{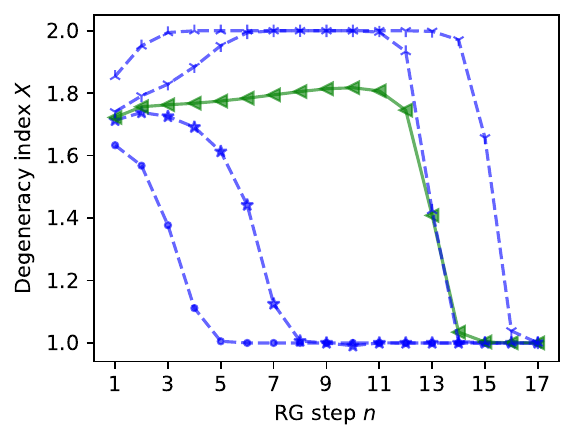}
}
~
\subfloat[The RG flows generated by the TRG]{
    \includegraphics[scale=0.8,
    valign=c]{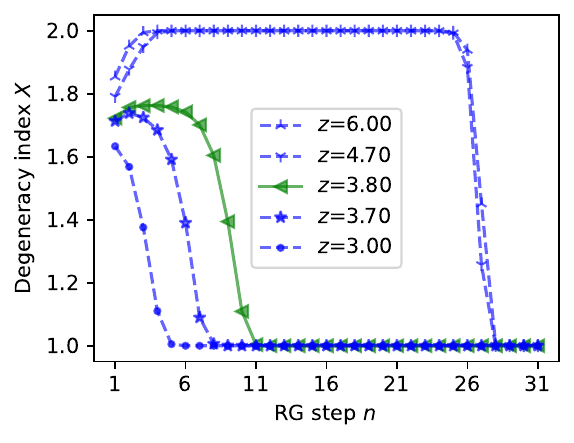}
}
\caption{\label{fig:model-zplusRGflows}
The RG flows of the degeneracy index $X$ near the positive-$z$ transition point $z_c^+$ generated by (a) the HOTRG and (b) the TRG at bond dimension $\chi=10$.
The solid line in the middle is the RG flow closest to the transition point $z_c^{+} \approx 3.80$.
Two lines below have $z < z_c^{+}$ and belong to the trivial phase, while two lines above have $z > z_c^+$ and belong to the SSB phase.
For the flows generated by the HOTRG starting with $z>z_c^+$, the tensor stays at the SSB fixed point ($X=2$) until the RG step $n\approx 11$.
The corresponding flows generated by the TRG are more stable, and the tensor starts to move away from the SSB fixed point after $n\approx 26$ RG steps.
    }
\end{figure*}

For the positive-$z$ transition, the relevant symmetries are lattice symmetries, specifically the rotation and reflection symmetries in Eqs.~\eqref{eq:model:A0rotsym} and~\eqref{eq:model:A0reflsym}.
The bond matrix $\sigma_0=1$ is trivial in Eq.~\eqref{eq:model:initialTN} so all the existing TNRG methods for a square-lattice tensor network are applicable.
We choose to compare the tensor RG flows generated by the TRG~\cite{Levin:Nave:2007} and the higher-order tensor renormalization group (HOTRG)~\cite{Xie:2012:hotrg}.
Although neither of the schemes incorporates lattice symmetries explicitly, they have quite distinct symmetry features.
The TRG splits the 4-leg tensor $A$ in Eq.~\eqref{eq:model:initialTN} diagonally using SVD, in a manner similar to Eq.~\eqref{eq:model:c2cp}.
Although the lattice symmetries of the tensor are not explicitly exploited, they are expected to be inherited in these SVD splittings\footnote{The precise meaning of this statement will become clear when we study the implications of lattice symmetries in TRG later in Sec.~\ref{sec:TNRG:symSVD}.
}.
The only source of the perturbations that break the lattice symmetries is due to machine precision in numerical calculations.
However, in the HOTRG, coarse grainings in $x$ and $y$ directions are performed in series, and the arbitrary choice of the order of the coarse graining explicitly breaks the lattice rotation symmetry, as well as the lattice reflection symmetries along $l_1,l_2$ axes (see Fig.~\ref{fig:model-zpground} and Eq.~\eqref{eq:model:A0reflsym}).
Therefore, it is expected that the HOTRG should break the lattice symmetry more severely than the TRG and that the RG flows generated by the HOTRG near the ordered-phase fixed point should be less stable than those generated by the TRG.
In Fig.~\ref{fig:model-zplusRGflows}, this expectation is checked numerically by plotting the RG flows of the degeneracy index $X$.

\begin{figure*}[tb]
\subfloat[A complex tensor-network representation]{
    \includegraphics[scale=0.8,
    valign=c]{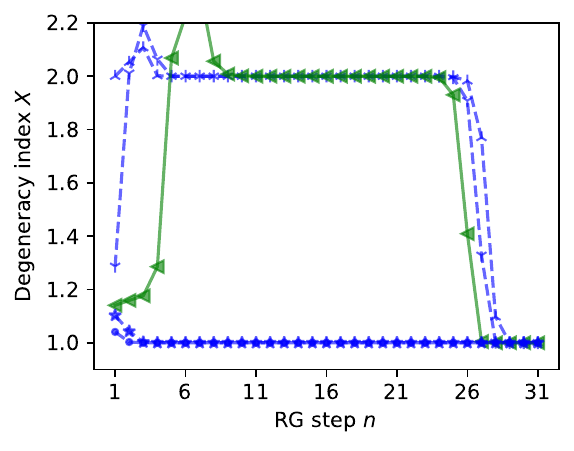}
}
~
\subfloat[A real tensor-network representation]{
    \includegraphics[scale=0.8,
    valign=c]{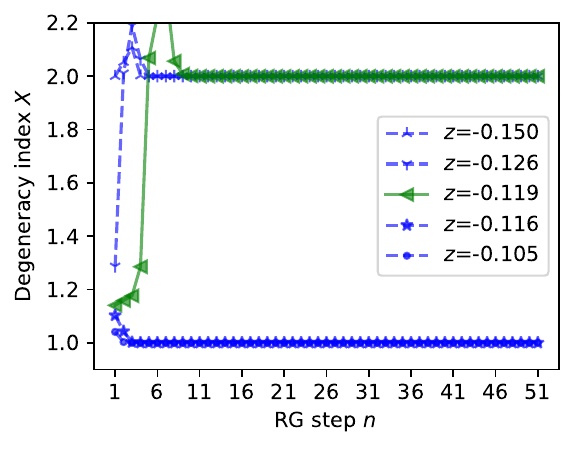}
}
\caption{\label{fig:model-znegRGflows}
The RG flows of the degeneracy index $X$ near the negative-$z$ transition point $z_c^-$ generated by the TRG at bond dimension $\chi=10$.
The solid line in the middle is the RG flow closest to the repulsive-core singularity $z_c^- \approx -0.119$.
Two lines below have $z>z_c^-$ and belong to the trivial phase, while two lines above have $z<z_c^-$ and belong to the SSB phase.
(a) The RG flow starts with a complex tensor-network representation of the model.
Perturbation due to machine-precision error breaks the $\mathcal{PT}$ symmetry and thus makes the SSB fixed point ($X=2$) unstable after $n\approx 26$ RG steps.
(b) The RG flow starts with a real tensor-network representation of the model.
Since the SVD splittings in the TRG make sure all tensors remain real, the $\mathcal{PT}$ symmetry is exactly preserved.
Therefore, the SSB fixed point ($X=2$) should be strictly stable.
Although the figure only shows up to the RG step $n=51$, we have checked that this stability persists even after $n=100$ RG steps.
    }
\end{figure*}

As for the negative-$z$ transition, the relevant symmetry is the $\mathcal{PT}$ symmetry that is manifested in the tensor network as the fact that the tensors are real-valued (see Eq.~\eqref{eq:model:A0real}).
Most TNRG schemes are not applicable to a tensor network with a nontrivial bond matrix $\sigma_0$ in Eq.~\eqref{eq:model:initialTN}.
However, this can be easily resolved after the bond matrix $\sigma_0$ is absorbed into the 4-leg tensor $A$ through its two legs pointing to the positive $x,y$ directions.
Since the SVD splits a real-valued matrix into real-valued pieces, the $\mathcal{PT}$ symmetry is strictly preserved in the TRG as long as the tensors in the initial tensor network are real-valued;
no machine-precision error can break this symmetry since the data type in numerical calculations remains real floating-point numbers.
Therefore, the SSB fixed point is strictly stable when the TRG is applied to the real tensor-network representation of the model in Eqs.~\eqref{eq:model:initialTN} and~\eqref{eqsub:model:A0}.
The numerical evidence of this is shown in Fig.~\ref{fig:model-znegRGflows}(b).

This will be contrasted with the RG flows when the TRG is applied to a complex tensor-network representation of the model.
The complex representation can be obtained by taking the square root of $\sigma_0$ for $z<0$: $\sqrt{\sigma_0} = \diag(1, i)$, and then acting each one of them on the $c_p$ tensor in Eq.~\eqref{eq:model:c2cp}.
The resultant square-lattice tensor network has no bond matrix and is composed of a different 4-leg tensor $\tilde{A}$ that is obtained by acting $\sqrt{\sigma_0}$ on the four legs of $A$ in Eq.~\eqref{eq:model:A0}.
The expectation is that errors due to machine precision should introduce perturbations in this complex presentation, making the SSB fixed point unstable.
We check this expectation numerically in Fig.~\ref{fig:model-znegRGflows}(a), where the SSB fixed point becomes unstable after $n\approx 26$ RG steps.

\section{Exploit the lattice and $\mathcal{PT}$ symmetries in TNRG with loop optimization\label{sec:TNRG}}
We propose a 2D TNRG transformation where lattice-reflection, lattice-rotation and $\mathcal{PT}$ symmetries are exploited and strictly imposed.
The proposed scheme is based on Levin and Nave's tensor renormalization group (TRG)~\cite{Levin:Nave:2007} and the loop optimization in the loop-TNR~\cite{loop-TNR:2017}.
The essential ingredients of our scheme are (i) writing down a proper definition of the symmetries for the coarse-grained tensor network and (ii) revealing the implication of the symmetries in the singular value decomposition (SVD) splitting step of the TRG.
Incorporating the loop optimization becomes straightforward after these two ingredients are ready.
We will also prove that both symmetries are preserved under the proposed RG transformation.
Our scheme is a generalization of the existing symmetric version of the loop-TNR~\cite{loop-TNR:2017}.


\subsection{Definition of the lattice and $\mathcal{PT}$ symmetries for the coarse-grained tensor network\label{sec:TNRG:defsym}}
We claim that the coarse-grained tensor network representation of the partition function has the following form:
\begin{align}
    \label{eq:scheme:initialTN}
    \mathcal{Z} = 
    \includegraphics[scale=0.9, valign=c]{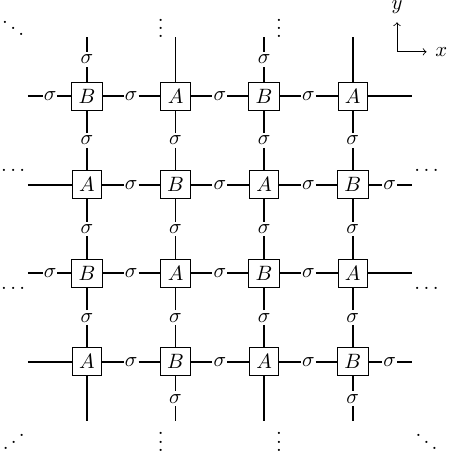},
\end{align}
where the bond matrix $\sigma$ is diagonal with its diagonal entries to be either $+1$ or $-1$.
In this coarse-grained tensor network, the definition of $\mathcal{PT}$ symmetry takes the same form as the initial tensor network of the hard-square model---all tensors in the tensor network are real-valued, $A=A^*, B=B^*, \sigma=\sigma^*$:
\begin{align}
    \label{eq:scheme:ABreal}
    \includegraphics[scale=1.0, valign=c]{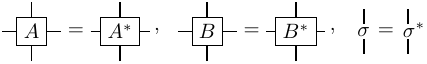},
\end{align}
where the complex conjugate implements the time-reversal operator $\mathcal{T}$ and the parity operator $\mathcal{P}$ is an identity matrix acting on each leg of a tensor in this real representation.

This tensor network representation has both the lattice-reflection and lattice-rotation symmetries.
These two symmetries can be represented by certain symmetry properties of the tensors $A$ and $B$; 
the symmetry properties have two suitable forms, one of which is stronger than the other.

\subsubsection{Weak form of the lattice symmetries}
The weak form is expressed in terms of how the tensor $B$ can be determined by several lattice symmetry operations of the tensor $A$:
\begin{align}
    \label{eq:scheme:weaksym}
    \includegraphics[scale=1.0, valign=c]{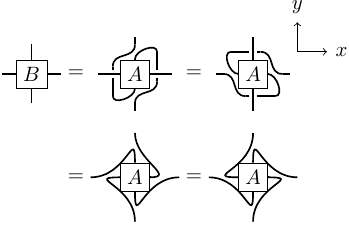}.
\end{align}
The first two equal signs are lattice reflections along the $x$ and $y$ axes, while the second two are $90^\circ$ lattice rotations counterclockwise and clockwise.
Notice that the above weak form in Eq.~\eqref{eq:scheme:weaksym} implies the following diagonal reflection symmetry of the tensor $A$ along $l_1$ and $l_2$ axes:
\begin{align}
    \label{eq:scheme:Arefll1l2sym}
    \includegraphics[scale=1.0, valign=c]{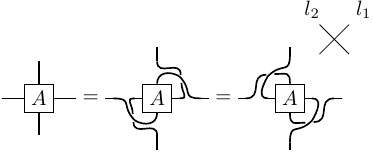},
\end{align}
and the same is true for the tensor $B$.
We will show in Secs.~\ref{sec:TNRG:symTRG} and~\ref{sec:TNRG:loopEF} that this weak form is satisfied by all 4-leg tensors during the coarser graining of the proposed RG transformation, including the 4-leg tensors in the intermediate step.

\subsubsection{Strong form of the lattice symmetries}
The strong form of the two lattice symmetries is represented by the following symmetry of the tensor $A$ itself, which involves a diagonal SWAP-gauge matrix $g$~\cite{Lyu:Kawashima:2025:reflsym}:
\begin{subequations}
\label{eq:scheme:strongsym}
\begin{align}
    \label{eq:scheme:reflstrong}
    \includegraphics[scale=1.0, valign=c]{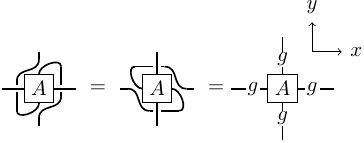}
\end{align}
for lattice reflections along the $x$ and $y$ axes,
\begin{align}
    \label{eq:scheme:rotstrong}
    \includegraphics[scale=1.0, valign=c]{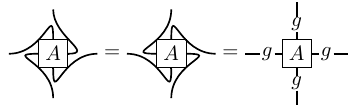}
\end{align}
for $90^\circ$ lattice rotation counterclockwise and clockwise, and
\begin{align}
    \label{eq:scheme:ABstrong}
    \includegraphics[scale=1.0, valign=c]{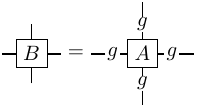}
\end{align}
for the relationship between tensor $B$ and tensor $A$.
\end{subequations}
In Eqs.~\eqref{eq:scheme:reflstrong} to~\eqref{eq:scheme:ABstrong}, just like the bond matrix $\sigma$, the SWAP-gauge matrix $g$ is also diagonal, with its diagonal entries to be either $+1$ or $-1$.
The initial tensor of the 1NN hard-square model satisfies this strong form with the SWAP-gauge matrix being identity.
We will explain how to determine the SWAP-gauge matrix $g$ for a coarse-grained tensor in Sec.~\ref{sec:TNRG:symSVD}.
Notice that the strong form in Eq.~\eqref{eq:scheme:strongsym} implies the weak form in Eq.~\eqref{eq:scheme:weaksym}.
We will see in Sec.~\ref{sec:TNRG:symTRG} that this strong form is satisfied only by the input and output tensors of the proposed RG transformation, but not necessarily the 4-leg tensors in the intermediate step.

\subsection{Symmetric SVD splitting\label{sec:TNRG:symSVD}}
In 2D, there are two basic TNRG schemes for performing the simple coarse graining: tensor renormalization group (TRG)~\cite{Levin:Nave:2007} and higher-order tensor renormalization group (HOTRG)~\cite{Xie:2012:hotrg}.
We choose the TRG since it has lower computational costs and a more isotropic way of deforming the tensor network, making it easier to incorporate the lattice symmetries.

In order to reveal the consequence of the lattice symmetries in the TRG, we propose a \emph{symmetric singular value decomposition (SVD) splitting} to replace the usual SVD splitting in the TRG.
We focus on the splitting of the tensor $A$ along the $l_1$ axis in Eq.~\eqref{eq:scheme:Arefll1l2sym}.
Due to the lattice-reflection symmetry along the $l_1$ axis, as is shown in the first equal sign in Eq.~\eqref{eq:scheme:Arefll1l2sym}, the tensor $A$ can be seen as a symmetric matrix along the $l_1$ axis.
According to the spectrum theorem in linear algebra, we can perform the following truncated eigendecomposition (ED):
\begin{subequations}
\begin{align}
    \label{eq:scheme:EDA}
    \includegraphics[scale=1.0, valign=c]{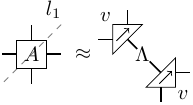},
\end{align}
where the diagonal matrix $\Lambda$ contains the eigenvalues, while the 3-leg tensor $v$ contains the eigenvectors.
The arrow in the diagram of $v$ indicates the order of the two legs of $v$.
Numerical truncations can occur in this step by ordering the eigenvectors according to the absolute value of their eigenvalues and keeping only the first $\chi$ largest ones in the decomposition.
Then, we define the sign of the eigenvalues as the new bond matrix $\sigma' = \sign{\Lambda}$ and distribute the absolute value of the eigenvalues to the two copies of the 3-leg tensor $v$:
\begin{align}
    \label{eq:scheme:weightsplit}
    \includegraphics[scale=1.0, valign=c]{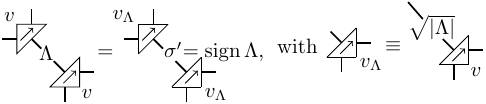},
\end{align}
where $\sqrt{|\Lambda|}$ means first taking the element-wise absolute value and then the element-wise square root.
We call the resultant truncated decomposition in Eqs.~\eqref{eq:scheme:EDA} and~\eqref{eq:scheme:weightsplit} a \emph{symmetric SVD splitting}:
\begin{align}
    \label{eq:scheme:symSVD}
    \includegraphics[scale=1.0, valign=c]{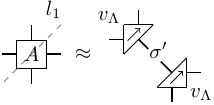}.
\end{align}
\end{subequations}

\begin{remark}
The symmetric SVD splitting in Eq.~\eqref{eq:scheme:symSVD} is closely related to the usual SVD splitting when the tensor satisfies the diagonal reflection symmetry in Eq.~\eqref{eq:scheme:Arefll1l2sym} because of the close relationship between the SVD and the ED for a real symmetric matrix.
We ignore the truncations for now.
It is straightforward to see that the element-wise absolute value of the diagonal matrix $|\Lambda|$ in the ED of Eq.~\eqref{eq:scheme:EDA} contains all the singular values in the SVD of the tensor $A$ along the same axis.
Meanwhile, in the SVD, the sign of the diagonal matrix $\sigma'=\sign{\Lambda}$ will be absorbed into either one of the copies of the tensor $v$ in Eq.~\eqref{eq:scheme:EDA} to give one orthogonal matrix and the other copy is the second orthogonal matrix.
Due to this direct correspondence, the truncations in the splitting happen in the same way for both the symmetric and usual SVD splitting.
\end{remark}

\begin{remark}
The $\mathcal{PT}$ symmetry is automatically preserved and imposed in the symmetric SVD splitting in Eq.~\eqref{eq:scheme:symSVD} for a real-valued tensor $A$ since both $v_\Lambda$ and $\sigma'$ are real-valued due to the spectrum theorem of a real symmetric matrix.
In numerical calculations, the eigenvalue solver for a real symmetric matrix should be used in the truncated ED in Eq.~\eqref{eq:scheme:EDA}.
\end{remark}

\begin{remark}
    The bond matrix $\sigma'$ arises naturally in the splitting of the tensor when the lattice symmetries are exploited without letting complex numbers appear.
\end{remark}

\begin{theorem}[Symmetry of the 3-leg tensor $v_{\Lambda}$ in symmetric SVD splitting]
\label{theo:vLsym}
When the tensor $A$ satisfies the diagonal reflection symmetry along both the $l_1$ and $l_2$ axes in Eq.~\eqref{eq:scheme:Arefll1l2sym}, the tensor $v_{\Lambda}$ in its symmetric SVD splitting along the $l_1$ axis in Eq.~\eqref{eq:scheme:symSVD} can be chosen to satisfy the following symmetry property:
\begin{align}
    \label{eq:scheme:vLsym}
    \includegraphics[scale=1.0, valign=c]{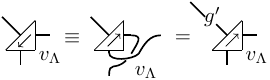},
\end{align}
where the first equal sign comes from the definition of the arrow in the diagrammatic representation of $v$: a change of the direction of arrow corresponds to a transposition of the two corresponding tensor legs.
The matrix $g'$, known as a SWAP-gauge matrix~\cite{Lyu:Kawashima:2025:reflsym}, is diagonal, with its diagonal entries being either $+1$ or $-1$.
\end{theorem}

\begin{proof}
    We only need to show that the 3-leg tensor $v$ in Eq.~\eqref{eq:scheme:EDA} satisfies the above symmetry property in Eq.~\eqref{eq:scheme:vLsym} because $v_{\Lambda}$ only differs by a diagonal matrix $\sqrt{|\Lambda|}$ acting on the diagonal leg of $v$ (see Eq.~\eqref{eq:scheme:weightsplit}) and two diagonal matrices commute with each other.
To this end, recall that the tensor $v$, with its diagonal leg taking a particular value, is an eigenvector of the tensor $A$ when $A$ is treated as a matrix across the $l_1$ axis.
Due to the diagonal reflection symmetry along the $l_2$ axis in Eq.~\eqref{eq:scheme:Arefll1l2sym}, this matrix commutes with the SWAP operator that implements the transposition of two legs of the tensor~\cite{Lyu:Kawashima:2025:reflsym}.
This implies that the set of eigenvectors of $A$ in the ED of Eq.~\eqref{eq:scheme:EDA} can be chosen to be the eigenvectors of the SWAP operator.
Since the SWAP operator squares to the identity operator, its eigenvalues can be either $+1$ or $-1$.
The SWAP-gauge matrix $g'$ in Eq.~\eqref{eq:scheme:vLsym} encodes eigenvalues of the SWAP operator; therefore, $g'$ is diagonal, with the diagonal entries being either $+1$ or $-1$.
\end{proof}

\begin{remark}
Only the weak form of the lattice symmetries in Eq.~\eqref{eq:scheme:weaksym} is needed for performing the symmetric SVD splitting in Eq.~\eqref{eq:scheme:symSVD} and for the symmetry property of the 3-leg tensor $v_{\Lambda}$ in Eq.~\eqref{eq:scheme:vLsym}.
\end{remark}

\subsection{TRG with the lattice and $\mathcal{PT}$ symmetries\label{sec:TNRG:symTRG}}
We are ready to demonstrate how to perform the TRG where both lattice and $\mathcal{PT}$ symmetries are preserved and imposed.
The proposed symmetric TRG consists of three steps.
The first step is fixing a representation for the tensor network of the partition function in Eq.~\eqref{eq:scheme:initialTN} using a \emph{rotation trick}.
The second step is the symmetric SVD splitting of the tensors into 3-leg and new bond tensors.
The third step is the recombination of 3-leg tensors into new 4-leg tensors.

\emph{Step 1: the rotation trick. ---}
Due to lattice symmetries, there is certain freedom in the tensor network representation in Eq.~\eqref{eq:scheme:initialTN}.
For example, the tensor $B$ can be obtained through several symmetry operations of the tensor $A$ according to Eq.~\eqref{eq:scheme:weaksym}.
We will choose a representation convenient for exploiting lattice symmetries.
Use $A_{\frac{\pi}{2}}, A_{\pi}, A_{\frac{3\pi}{2}}$ to denote counterclockwise $\pi/2, \pi, 3\pi/2$ rotation of the 4-leg tensor $A$:
\begin{align}
    \label{eq:scheme:Arotdef}
    \includegraphics[scale=0.9, valign=c]{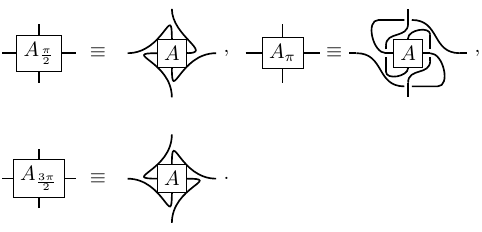}
\end{align}
Due to the weak form of the lattice symmetries in Eq.~\eqref{eq:scheme:weaksym}, it is easy to see that $B = A_{\frac{\pi}{2}} = A_{\frac{3\pi}{2}}$ and $A = A_{\pi}$.
Then, we use copies of the tensor $A$ and its rotation to rewrite the partition function Eq.~\eqref{eq:scheme:initialTN} in the following form:
\begin{align}
    \label{eq:scheme:fixTNrep}
    \mathcal{Z} = 
    \includegraphics[scale=1.0, valign=c]{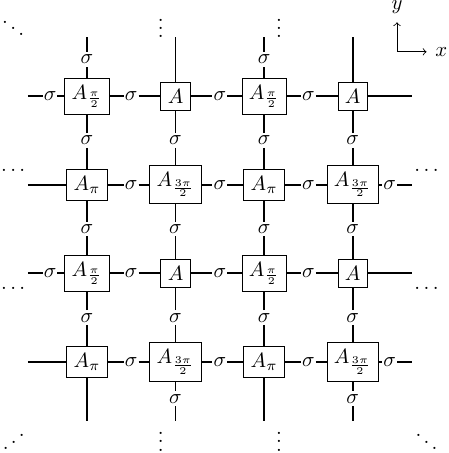}.
\end{align}

\emph{Step 2: the symmetric SVD splitting. ---}
Apply the symmetric SVD splitting in Eq.~\eqref{eq:scheme:symSVD} to the tensor $A$ and its rotations in Eq.~\eqref{eq:scheme:fixTNrep}.
The resultant tensor network becomes
\begin{align}
    \label{eq:scheme:SVDsplitTN}
    \includegraphics[scale=1.0, valign=c]{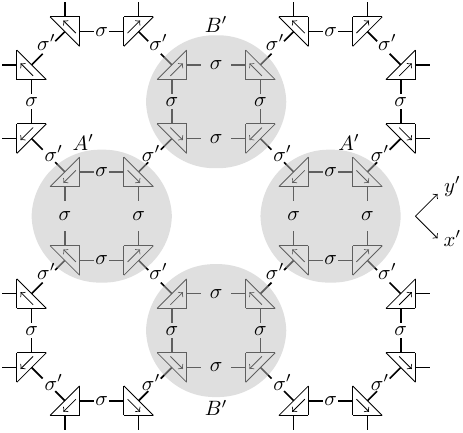}.
\end{align}

\emph{Step 3: recombination. ---}
Define the new tensors $A',B'$ by contracting copies of the 3-leg tensor $v_{\Lambda}$ and old bond matrix $\sigma$ as is indicated in Eq.~\eqref{eq:scheme:SVDsplitTN}.
The new tensor network consists of $A', B'$ and $\sigma'$, and is $45^\circ$ tilted; the axes $x',y'$ are chosen according to the indication in Eq.~\eqref{eq:scheme:SVDsplitTN}.
The new tensors $A'$ and $B'$ satisfy the weak form of the lattice symmetries in Eq.~\eqref{eq:scheme:weaksym} automatically, thanks to the rotation trick in Step 1.
Moreover, the symmetry of the tensor $v_{\Lambda}$ in Theorem~\ref{theo:vLsym} also indicates that tensors $A'$ and $B'$ also satisfy the strong form of the lattice symmetries in Eq.~\eqref{eq:scheme:strongsym}.

\begin{remark}
    We want to point out that Step 1 and Step 2 of the proposed symmetric TRG only require the weak form of the lattice symmetries in Eq.~\eqref{eq:scheme:weaksym}.
    The resultant tensor network after Step 3 has the strong form of the lattice symmetries in Eq.~\eqref{eq:scheme:strongsym}.
\end{remark}

\begin{remark}
    The strong form of the lattice symmetries is not only preserved but also imposed in the proposed symmetric TRG.
    To see this, first notice that the weak form of the lattice symmetries is imposed since it is satisfied by $A',B'$ by the very structure of the tensor network in Eq.~\eqref{eq:scheme:SVDsplitTN}.
    Since the symmetry of the 3-leg tensor $v_{\Lambda}$ in Theorem~\ref{theo:vLsym} only needs the weak form of the lattice symmetries, it follows that the strong form is also imposed.
\end{remark}

\begin{remark}
    The TRG transformation has an RG rescaling factor $b = \sqrt{2}$.
    In this paper, we regard two TRG transformations as a single RG step with $b=2$:
    \begin{align}
    \label{eq:scheme:oneRGsimple}
        A,B,\sigma \xrightarrow[]{\text{TRG}}
        A',B',\sigma' \xrightarrow[]{\text{TRG}}
        A'',B'',\sigma''.
    \end{align}
    We call $A,B,\sigma$ original tensors, $A'',B'',\sigma'$ coarse-grained tensors, and $A',B',\sigma'$ intermediate tensors.
\end{remark}

\begin{remark}
    The proposed symmetric TRG is equivalent to the usual TRG if the machine precision in numerical calculation can be ignored.
    Without diving into the detailed proof, we sketch the reason below.
    The tensor network in Eq.~\eqref{eq:scheme:initialTN} can be made to consist of only one type of tensor $A_{\text{uni}}$ by properly applying the bond matrix $\sigma$ and SWAP-gauge matrix $g$ to the tensor $A$ (see Appendix~\ref{app:buildTM}).
    Afterwards, the usual TRG can be applied, where the 3-leg tensor in the SVD splitting can be determined by applying $\sigma,g$, and $\sigma'$ to $v_{\Lambda}$ in the symmetric SVD splitting.
    Therefore, the truncations in the TRG happen in the same way as in the proposed symmetric TRG.
\end{remark}


\subsection{Incorporating the loop optimization\label{sec:TNRG:loopEF}}
After demonstrating how to perform the simple coarse graining using the symmetric TRG, we now present one way to incorporate entanglement filtering (EF) into one RG step in Eq.~\eqref{eq:scheme:oneRGsimple}.
There are several suitable EF schemes available in 2D~\cite{loop-TNR:2017,Hauru:2018,Evenbly:2018:fet}; 
all of them need certain generalizations to deal with the tensor network in Eq.~\eqref{eq:scheme:initialTN}, where lattice symmetries are present and there is a nontrivial bond matrix.
We choose the loop optimization idea in the loop-TNR~\cite{loop-TNR:2017} since it is easiest to generalize for incorporating lattice symmetries in a tensor network with bond matrices\footnote{In our attempt to generalize the schemes in Ref.~\cite{Hauru:2018,Evenbly:2018:fet}, it seems unavoidable to modify the bond matrices in the optimization.
Intuitively, it feels unnatural to change bond matrices in the optimization since they carry the sign of eigenvalues of some tensor and have a rigid structure.
}.

Using the corner-double line (CDL) tensor~\cite{Gu:Wen:2009} as the input of the simple RG transformation in Eq.~\eqref{eq:scheme:oneRGsimple}, one can determine the location of the redundant entanglement in the tensor network in Eq.~\eqref{eq:scheme:initialTN}, which is indicated as thick-stroke loops in the following diagram:
\begin{align}
    \label{eq:scheme:looploc}
    \includegraphics[scale=1.0, valign=c]{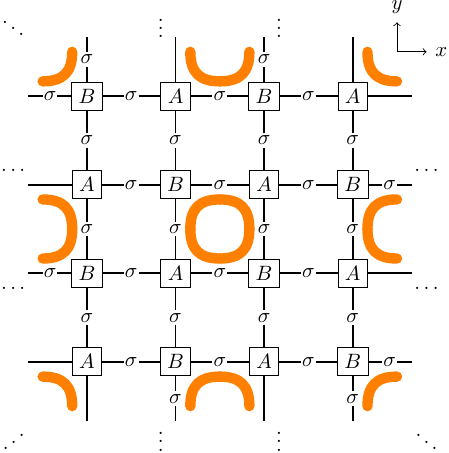}.
\end{align}

These loops can be filtered by incorporating an additional \emph{loop optimization} step between Step 2 and Step 3 of the first TRG in the RG step.
By comparing Eq.~\eqref{eq:scheme:fixTNrep} and Eq.~\eqref{eq:scheme:SVDsplitTN}, we write down the approximation of this loop optimization:
\begin{align}
    \label{eq:scheme:loopopt}
    \includegraphics[scale=1.0, valign=c]{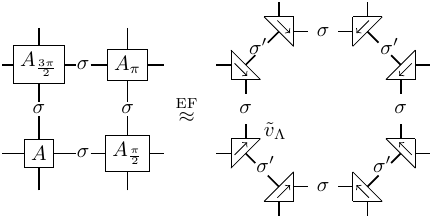},
\end{align}
where $\sigma'$ is the same as the bond matrix in Eq.~\eqref{eq:scheme:SVDsplitTN} and Eq.~\eqref{eq:scheme:symSVD}, and the 3-leg tensor $\tilde{v}_{\Lambda}$ is initialized as $v_{\Lambda}$ in the symmetric SVD splitting in Eq.~\eqref{eq:scheme:symSVD}.
In the loop optimization in Eq.~\eqref{eq:scheme:loopopt}, the tensors $A, \sigma, \sigma'$ remain fixed while the 3-leg tensor $\tilde{v}_{\Lambda}$ is varied to make the tensor network on the right-hand side a good approximation to that on the left-hand side. 
The optimization method developed in Refs.~\cite{loop-TNR:2017,Evenbly:2018:fet,Lyu:Kawashima:2025:reflsym} is adopted to determine the optimal $\tilde{v}_{\Lambda}$ iteratively.
A self-contained review of this strategy and its subtleties can be found in Appendix~\ref{app:loopopt}.
After determining the optimal $\tilde{v}_{\Lambda}$, new tensors $A', B'$ are obtained by contracted copies of $\tilde{v}_{\Lambda}$ and $\sigma$ as is indicated in Eq.~\eqref{eq:scheme:SVDsplitTN}.
In the second half of the RG transformation, the symmetric TRG can be applied as usual.
Therefore, after incorporating the loop optimization, a single RG step of the proposed method is
\begin{align}
    \label{eq:scheme:oneRGloop}
    A,B,\sigma \xrightarrow[+\text{loop-opt}]{\text{TRG}}
    A',B',\sigma' \xrightarrow[]{\text{TRG}}
    A'',B'',\sigma''.
\end{align}

\begin{remark}
    Although an additional step is added, there is still only one approximation in the first half of the RG step.
    This is because we can treat the symmetric SVD splitting as the initialization of the loop optimization.
    Therefore, the only approximation in the first half of the RG step is the loop approximation in Eq.~\eqref{eq:scheme:loopopt}.
\end{remark}

\begin{remark}
    After incorporating the loop optimization, the 3-leg tensor $\tilde{v}_{\Lambda}$ no longer has the symmetry in Eq.~\eqref{eq:scheme:vLsym}, which is a property of the symmetric SVD splitting.
    Therefore, the tensor network consisting of copies of tensors $A',B',\sigma'$ does not have the strong form of the lattice symmetries in Eq.~\eqref{eq:scheme:strongsym}.
    However, the symmetric TRG can still be applied since the tensor network still has the weak form of the lattice symmetries.
    The final coarse-grained tensor network consisting of copies of tensors $A'',B'', \sigma'$ has the strong form of the lattice symmetries.
    The $\mathcal{PT}$ symmetry is also preserved and imposed since tensors in every step of the RG transformation remain real-valued.
\end{remark}


\subsection{Relationship with the symmetric loop-TNR\label{sec:TNRG:relationLTNR}}
The proposed symmetric TRG with loop optimization is a generalization of the symmetric version of the (slightly modified) loop-TNR briefly mentioned in an appendix of Ref.~\cite{loop-TNR:2017}.
The proposed scheme skips another entanglement filtering step in loop-TNR, where a change of basis is made on the bond of the tensor network with no approximation involved.
The main purpose of this step in loop-TNR is to filter out the CDL tensors.
However, we discover that the loop optimization in Eq.~\eqref{eq:scheme:loopopt} can reliably filter out the CDL tensors when certain subtleties are taken care of regarding the optimization method of the 3-leg tensor $\tilde{v}_{\Lambda}$ (see Appendix~\ref{app:loopopt}).
Furthermore, the loop-TNR applies loop optimization for every TRG transformation, while the proposed RG only performs loop optimization for every two TRG steps since it is the minimal amount needed to simplify the CDL tensors completely.

In the symmetric loop-TNR~\cite{loop-TNR:2017}, the new bond matrix $\sigma'$ in the loop optimization in Eq.~\eqref{eq:scheme:loopopt} is assumed to be trivial.
In our proposed method, we allow nontrivial $\sigma'$ with $-1$ in its diagonal entries.
Recall that the form of the approximation in Eq.~\eqref{eq:scheme:loopopt} is a natural consequence of the symmetric SVD splitting in Eq.~\eqref{eq:scheme:symSVD}.
The proposed method ensures that the symmetric SVD splitting provides a good initialization of the 3-leg tensor $\tilde{v}_{\Lambda}$ for the loop optimization in Eq.~\eqref{eq:scheme:loopopt}, whose approximation error is controlled by that of the symmetric SVD splitting.
With $\sigma'=1$ in the loop-TNR, the initial 3-leg tensor obtained from the usual SVD splitting might not give a good initialization for the loop optimization when the coarse-grained tensors of a model have negative eigenvalues in the ED across a diagonal axis (see Eq.~\eqref{eq:scheme:EDA}).

Here are pieces of numerical evidence supporting the claim above.
For the 2D Ising model, all bond matrices $\sigma,\sigma',\sigma''$ in a single RG step in Eq.~\eqref{eq:scheme:oneRGloop} are the identity matrix, which explains why the symmetric loop-TNR works well when the 2D Ising model is used as the benchmark model.
However, for both the positive-$z$ and negative-$z$ transition points of the 1NN hard-square model, the bond matrix $\sigma'$ in Eq.~\eqref{eq:scheme:oneRGloop} is nontrivial and contains $-1$ in its diagonal entries.
Therefore, we conjecture that the loop approximation in the symmetric loop-TNR in Ref.~\cite{loop-TNR:2017} might not work well on this model.

\section{Numerical results\label{sec:numres}}
As was explored in Sec.~\ref{sec:model:symPhase} and~\ref{sec:model:numdemo}, the symmetries that are spontaneously broken are the lattice symmetries for the positive-$z$ transition and the $\mathcal{PT}$ symmetry for the negative-$z$ transition.
Moreover, these symmetries are responsible for the stability of the RG fixed point corresponding to the phase that breaks the symmetry spontaneously (the SSB phase).
Since the proposed TNRG transformation in Sec.~\ref{sec:TNRG} incorporates the lattice and the $\mathcal{PT}$ symmetries, the fixed points of the SSB phases for both the positive $z>z_c^+$ and negative $z<z_c^-$ become strictly stable\footnote{
We have checked this numerically.
}.
We use the RG flows of the degeneracy index $X$ defined in Eq.~\eqref{eq:model:degIndDef} to diagnose the phase of the model at a given activity $z$ (see Figs.~\ref{fig:model-zplusRGflows} and~\ref{fig:model-znegRGflows}).
A bisection method~\cite{Lyu:Xu:Kawashima:2021}, also known as a ``shooting method''~\cite{Ebel:2025:newton}, can then be employed to estimate the critical activity for both transitions.
Afterwards, the tensor RG flow at the estimated activity is generated, and those tensors that are near to the critical fixed points are used to construct transfer matrices and extract scaling dimensions~\cite{Gu:Wen:2009,Ebel:2025:LDO}.
In Appendix~\ref{app:buildTM}, we will explain how to build a transfer matrix for a tensor network shown in Eq.~\eqref{eq:scheme:initialTN}.
The numerical results in this section can be reproduced using the {\tt Python} codes published at Ref.~\cite{Lyu:algo:lattGas}.

\subsection{Ising transition for positive activity\label{sec:numres:posz}}
As has been explained in Sec.~\ref{sec:model:transitions}, the positive-activity transition of the 1NN hard-square model belongs to the 2D Ising universality class.
For the 2D Ising Conformal Field Theory (CFT), the entanglement entropy of the density matrix of a subsystem scales logarithmically with the linear size $L$ of the subsystem: 
\begin{align}
    \label{eq:numres:EEscale}
    S(L) = \frac{c}{3}\log(L), 
\end{align}
where $c$ is the central charge of the CFT~\cite{Holzhey:1994,Vidal:2003,Calabrese:2009}.
According to an argument using this entanglement entropy~\cite{Levin:Nave:2007,Lyu:Kawashima:2024area}, this scaling implies a growth of RG truncation errors with respect to the RG step for any TNRG scheme without the EF.
A successful EF scheme is expected to tame the growth of RG errors at criticality.
As a direct check of whether the loop optimization proposed in Sec.~\ref{sec:TNRG:loopEF} fulfills this expectation, we compare the flow of RG errors of the proposed method with those of the TRG~\cite{Levin:Nave:2007} in Fig.~\ref{fig:zpos-RGerrs}, where the truncation error of the second TRG transformation is used as a measure of the RG error in one RG step.
The loop optimization not only tames the growth but also reduces the RG error significantly;
this implies a boost of the accuracy of the estimates of the critical activity and the scaling dimensions at the criticality.

\begin{figure}[tb]
\centering
\subfloat[The usual TRG (without EF)]{
    \includegraphics[width=1.0\columnwidth]{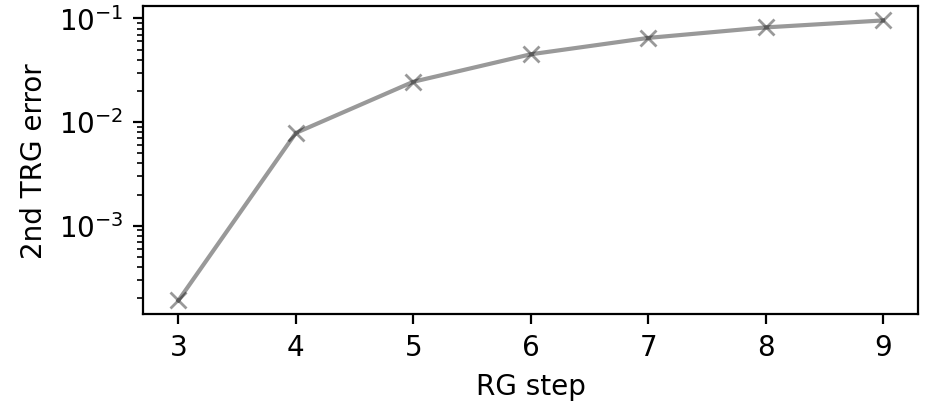}
}
\par
\subfloat[The proposed RG map (with loop optimization)]{
    \includegraphics[width=1.0\columnwidth]{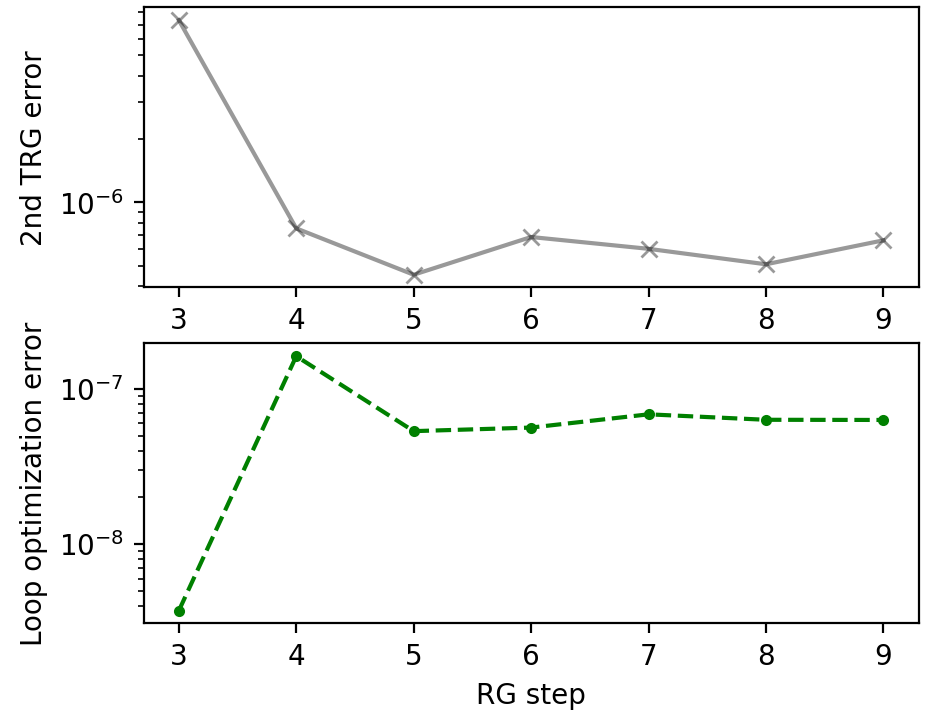}
}
\caption{\label{fig:zpos-RGerrs}
Flows of the RG errors of the TRG and the proposed scheme at the estimated critical activity $z_c^{+}$ for the positive-$z$ transition in Table~\ref{tab:estzpc}.
The bond dimension is $\chi=20$.
The growth of the RG error is tamed by the loop optimization.
(a) Without any EF process, the truncation error of the second TRG in the usual TRG map (a single RG step consists of two TRG transformations) grows to about $10^{-2}$ after four RG steps.
(b) By incorporating the loop optimization between the first and the second TRG (see Eq.~\eqref{eq:scheme:oneRGloop}), the truncation error of the second TRG in the proposed RG map stops growing and converges to less than $10^{-6}$ after about four RG steps.
The loop optimization error $\epsilon_{\text{loop}}^2 = 1 - F$ also converges to less than $10^{-7}$ 
(see Appendix~\ref{app:loopopt} for the definition of this error).
    }
\end{figure}

In Table~\ref{tab:estzpc}, we compare the estimates of $z_c^{+}$ using the proposed scheme and the usual TRG for bond dimensions $6\leq \chi \leq 20$.
For both methods, there is a clear trend of the improvement of the accuracy when the bond dimension $\chi$ increases;
however, the improvement is not monotonic.
At the same bond dimension $\chi$, the relative error of the value estimated by the proposed scheme is usually two orders of magnitude smaller than that estimated by the TRG.
In order to reach the accuracy of the proposed scheme at bond dimension $\chi=10$, the usual TRG needs to have $\chi=50$.

\begin{table}[tb]
\caption{\label{tab:estzpc}%
Estimates of the critical activity $z_c^{+}$ for positive-$z$ transition at various bond dimensions $\chi$ using the proposed scheme and the usual TRG without lattice symmetries.
The relative error is calculated by comparing our estimates to the estimated value $3.79625517391234(4)$ in Ref.~\cite{Guo:Blote:2002}.
}
\begin{ruledtabular}
\begin{tabular}{cll}
$\chi$ & the proposed scheme (error) & the TRG (error)\\
\colrule
6 & 3.810339 ($3 \times 10^{-3}$ ) & 3.7640 ($8 \times 10^{-3}$) \\
10 & 3.796440 ($5 \times 10^{-5}$) & 3.8096 ($2 \times 10^{-3}$) \\
12 & 3.796252 ($7 \times 10^{-7}$) & 3.7912 ($1 \times 10^{-3}$)\\
16 & 3.796232 ($6 \times 10^{-6}$) & 3.7902 ($2 \times 10^{-3}$) \\
18 & 3.796142 ($3 \times 10^{-5}$) & 3.7966 ($8 \times 10^{-5}$) \\
20 & 3.796245 ($3 \times 10^{-6}$) & 3.7976 ($3 \times 10^{-4}$) \\
30 &  & 3.7958 ($1 \times 10^{-4}$)\\
50 &  & 3.79619 ($2 \times 10^{-5}$)
\end{tabular}
\end{ruledtabular}
\end{table}

After the critical activity $z_c^{+}$ is estimated, we can generate a tensor RG flow that approaches the critical fixed point.
When the tensor is near the fixed point, we use two copies of the tensor to build the transfer matrix and extract scaling dimensions.
Similar to the estimates of $z_c^{+}$, we observe a clear trend of the improvement of the accuracy when the bond dimension $\chi$ increases for both methods.
In Fig.~\ref{fig:zpos-scDim}, we show the estimates of the scaling dimension with respect to the RG step using the usual TRG and the proposed scheme, both at their largest bond dimension, as is shown in Table~\ref{tab:estzpc}.
In order to reach a similar accuracy of the proposed method at $\chi=20$, the TRG needs to reach $\chi=50$.
Moreover, the estimates of scaling dimension are stable with respect to the RG step in the proposed method, while those estimated by the TRG drift slowly.

\begin{figure}[tb]
\centering
\subfloat[The usual TRG with $\chi=50$]{
    \includegraphics[width=1.0\columnwidth]{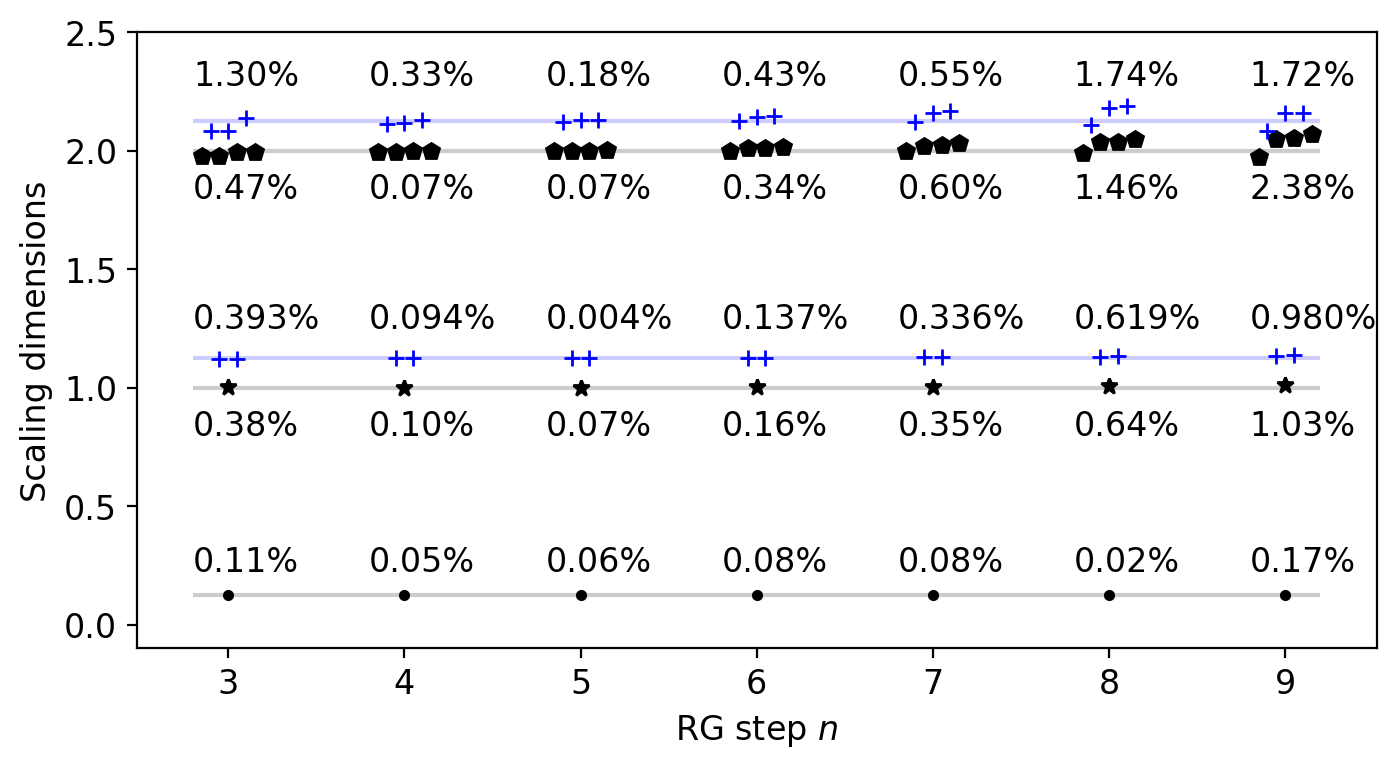}
}
\par
\subfloat[The proposed RG map with $\chi=20$]{
    \includegraphics[width=1.0\columnwidth]{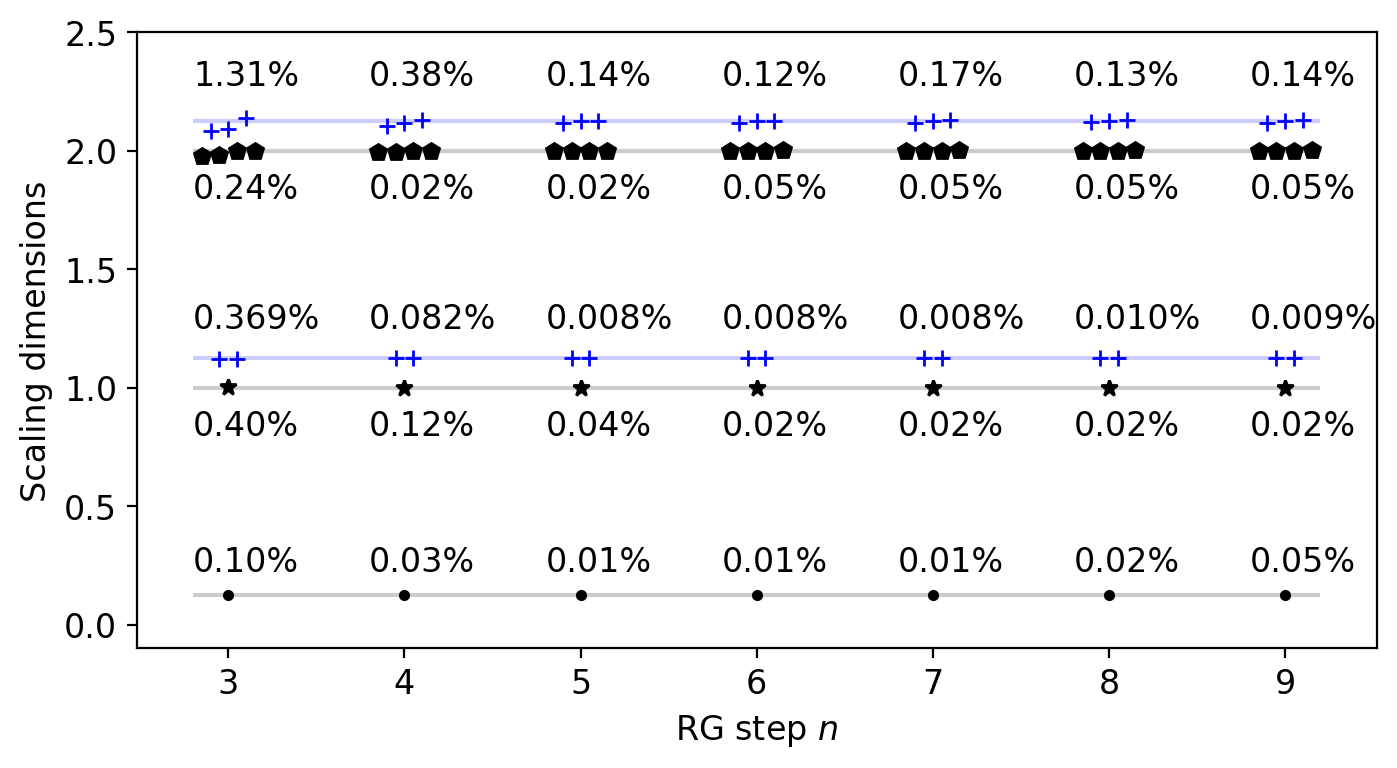}
}

\caption{\label{fig:zpos-scDim}
Estimates of the scaling dimensions at the positive-$z$ criticality using (a) the usual TRG and (b) the proposed method.
The tensor RG flow is generated at the estimated $z_c^{+}$ in Table~\ref{tab:estzpc}.
The horizontal lines are exact values of the 2D Ising CFT, while data points are the numerical estimates.
The percentages near the data points are the relative errors (we take the geometric average for the errors of degenerate scaling dimensions).
    }
\end{figure}

\subsection{Repulsive-core singularity for negative activity\label{sec:numres:negz}}
As has been explained in Sec.~\ref{sec:model:transitions}, the repulsive-core singularity at the negative activity $z_c^{-}$ belongs to the universality class of the Yang-Lee edge singularity, which is the simplest non-unitary CFT with a central charge $c=-22/5<0$.
The standard argument for the validity of the TNRG using the entanglement-entropy area law fails for such non-unitary CFT.
The reason is that the scaling of entanglement entropy in Eq.~\eqref{eq:numres:EEscale} for a non-unitary theory is based on a generalized reduced density matrix built from a combination of right and left ground states~\cite{Shimizu-Kawabata:2025}.
However, the main tensor in TNRG after many RG steps is interpreted as the ground state by evolving the system in the radial direction from the origin~\cite{Levin:Nave:2007,Lyu:Kawashima:2024area}.
This indicates that the reduced density matrix used for the truncations in the TNRG has the usual definition and is always positive semi-definite.
Therefore, we can no longer predict the behavior of the RG truncation errors with respect to the RG step in the TNRG, with or without the EF process.

However, we can still use the same strategy as the positive-$z$ transition in Sec.~\ref{sec:numres:posz} to check the effectiveness of the proposed method numerically.
We plot the RG errors of the
proposed method and the TRG in Fig.~\ref{fig:zneg-RGerrs}.
Similar to the criticality at the positive-$z$ Ising transition in Fig.~\ref{fig:zpos-RGerrs}(a), the RG error of the TRG grows to $10^{-1}$ after seven RG steps.
This RG error is reduced significantly to around $10^{-7}$ after incorporating the loop optimization.
This is a piece of numerical evidence showing that the idea of the EF can be effective for a criticality belonging to a non-unitary CFT, though the entanglement that is filtered in the EF cannot be what is in the area law formula in Eq.~\eqref{eq:numres:EEscale}.

\begin{figure}[tb]
\centering
\subfloat[The usual TRG (without EF)]{
    \includegraphics[width=1.0\columnwidth]{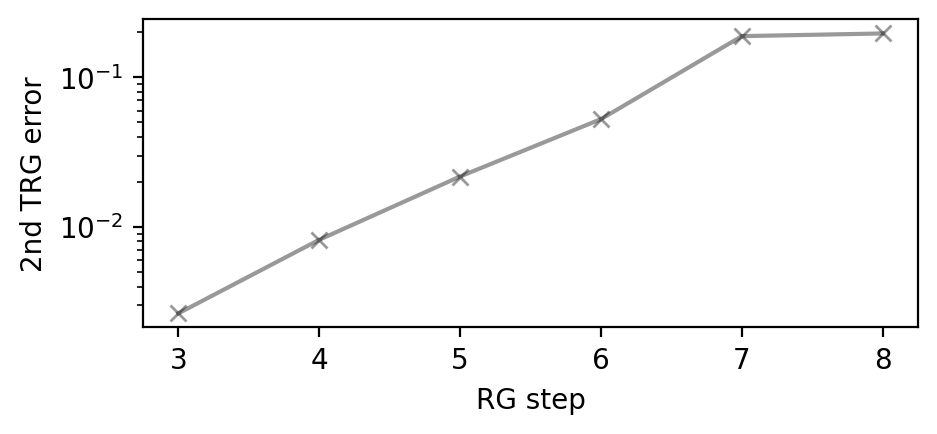}
}
\par
\subfloat[The proposed RG map (with loop optimization)]{
    \includegraphics[width=1.0\columnwidth]{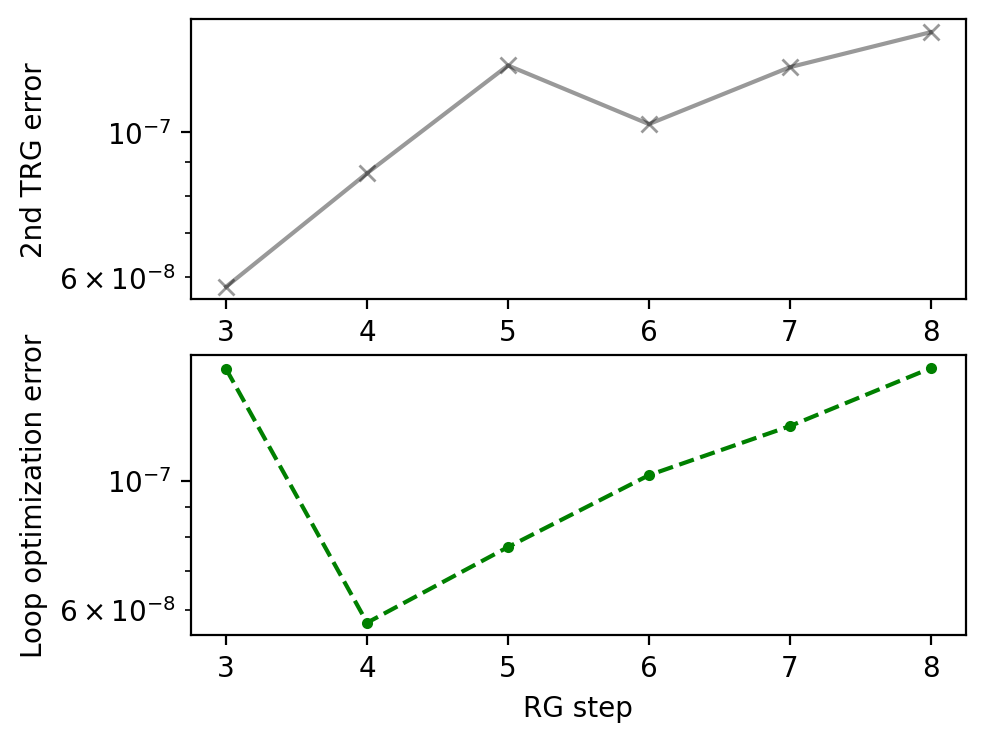}
}
\caption{\label{fig:zneg-RGerrs}
Flows of the RG errors of the TRG and the proposed scheme at the estimated critical activity $z_c^{-}$ for the negative-$z$ transition in Table~\ref{tab:estznc}.
The bond dimension is $\chi=10$.
(a) Without any EF process, from RG step $n=3$ to $n=7$, the truncation error of the second TRG in the usual TRG map (a single RG step consists of two TRG transformations) grows exponentially from $\sim 10^{-3}$ to $\sim 10^{-1}$.
(b) By incorporating the loop optimization between the first and the second TRG (see Eq.~\eqref{eq:scheme:oneRGloop}), the truncation error of the second TRG in the proposed RG map is reduced to $\sim 10^{-7}$ and grows much slower with respect to the RG step.
}
\end{figure}

In Table~\ref{tab:estznc}, we compare the estimates of $z_c^{-}$ using the proposed scheme and the usual TRG for bond dimensions $10\leq \chi \leq 20$.
The relative error of the estimate obtained by the proposed method is usually one to three orders of magnitude smaller than that estimated by the TRG at the same bond dimension $\chi$.
The proposed method with bond dimensions $10 \leq \chi \leq 20$ produces more accurate estimates than the usual TRG at $\chi=50$.
For both methods, the improvement of the estimate accuracy with the bond dimension $\chi$ is not as clear as the positive-$z$ transition in Table~\ref{tab:estzpc}, but the trend is there when the increment of $\chi$ is larger than $10$.

\begin{table}[tb]
\caption{\label{tab:estznc}%
Estimates of the critical activity $z_c^{-}$ for the negative-$z$ transition at various bond dimensions $\chi$ using the proposed scheme and the usual TRG.
The relative error is calculated by comparing our estimates to the estimated value $-0.11933888188(1)$ in Ref.~\cite{Todo:1999}.
}
\begin{ruledtabular}
\begin{tabular}{cll}
$\chi$ & the proposed scheme (error) & the TRG (error)\\
\colrule
10 & $-0.1193391 \,(2 \times 10^{-6})$ & $-0.119314 \,(2 \times 10^{-4})$ \\
12 & $-0.1193372 \,(1 \times 10^{-5})$ & $-0.119358 \,(2 \times 10^{-4})$ \\
16 & $-0.1193384 \,(4 \times 10^{-6})$ & $-0.119362 \,(2 \times 10^{-4})$ \\
20 & $-0.119338886 \,(4 \times 10^{-8})$ & $-0.119341 \,(2 \times 10^{-5})$ \\
30 &  & $-0.119337 \,(2 \times 10^{-5})$\\
50 &  & $-0.1193384 \,(4 \times 10^{-6})$
\end{tabular}
\end{ruledtabular}
\end{table}

The non-unitary CFT that describes the Yang-Lee edge singularity has only one RG-relevant primary field $\varphi$ with a scaling dimension $x_{\varphi}=-2/5$.
For the convenience of plotting the scaling dimensions, we shift all scaling dimensions by $2/5$.
Therefore, the scaling dimensions according to the 2D CFT are 
\begin{align}
    \label{eq:numres:scDznegExp}
    0, 0.4, 1, 1, 2, 2, 2, 2.4, 2.4, \ldots,
\end{align}
where $0$ and $0.4$ are for the primary field $\varphi$ and the identity, the doublet $1$ is for the first descendants $\partial_i \varphi$, the triplet $2$ is for the second descendants $\partial_i \partial_j \varphi$, and the doublet $2.4$ is for the energy-momentum tensor.
In numerical estimation, since we fix the smallest scaling dimension to $0$, the second smallest one reflects the accuracy of the estimate of the primary field $\varphi$.

Although the estimates of the critical activity of the negative-$z$ transition in Table~\ref{tab:estznc} have slightly better accuracy than that of the positive-$z$ transition in Table~\ref{tab:estzpc}, the estimation of the scaling dimensions seems to be more challenging for the negative-$z$ transition.
For the usual TRG with scaling dimensions $\chi \leq 20$, the estimate of the second and third smallest scaling dimensions $0.4$ and $1$ drift away very fast after three RG steps;
the higher ones drift even faster.
This drift becomes slower for larger bond dimensions.
Therefore, in Fig.~\ref{fig:zneg-scDim}(a), we choose to show the RG flow of the scaling dimensions at the largest bond dimension $\chi=50$ in Table~\ref{tab:estznc}.
Incorporating the loop optimization greatly improves the convergence of the scaling dimensions with respect to the RG step and the accuracy (see Fig.~\ref{fig:zneg-scDim}(b)).
Even at bond dimension $\chi=10$, the proposed method produces estimates of the scaling dimension with more accuracy and stability with respect to the RG step than the TRG at bond dimension $\chi=50$.
For bond dimensions $10<\chi \leq 20$, the RG flow of the scaling dimensions produced by the proposed method is similar to the $\chi=10$ results in Fig.~\ref{fig:zneg-scDim}(b), without much improvement of the accuracy.

\begin{figure}[tb]
\centering
\subfloat[The usual TRG with $\chi=50$]{
    \includegraphics[width=1.0\columnwidth]{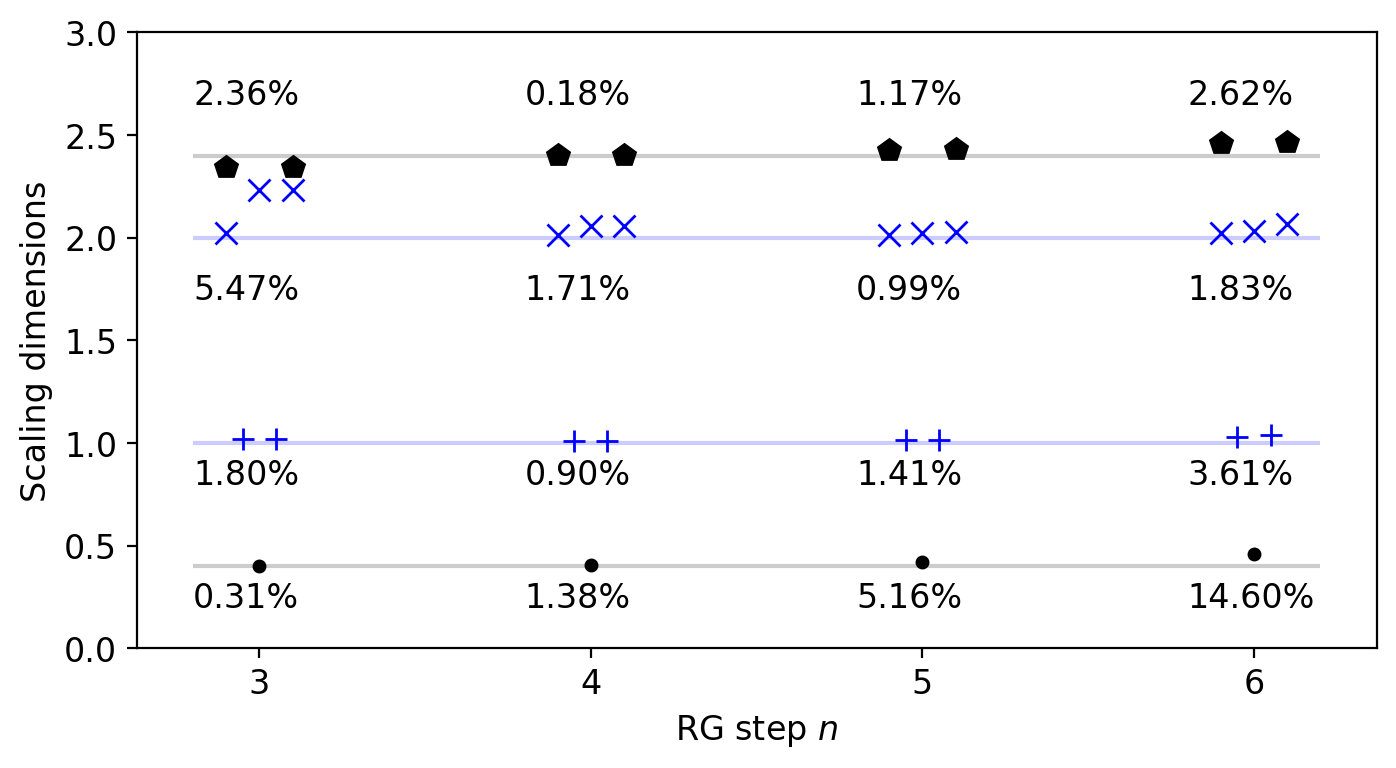}
}
\par
\subfloat[The proposed RG map with $\chi=10$]{
    \includegraphics[width=1.0\columnwidth]{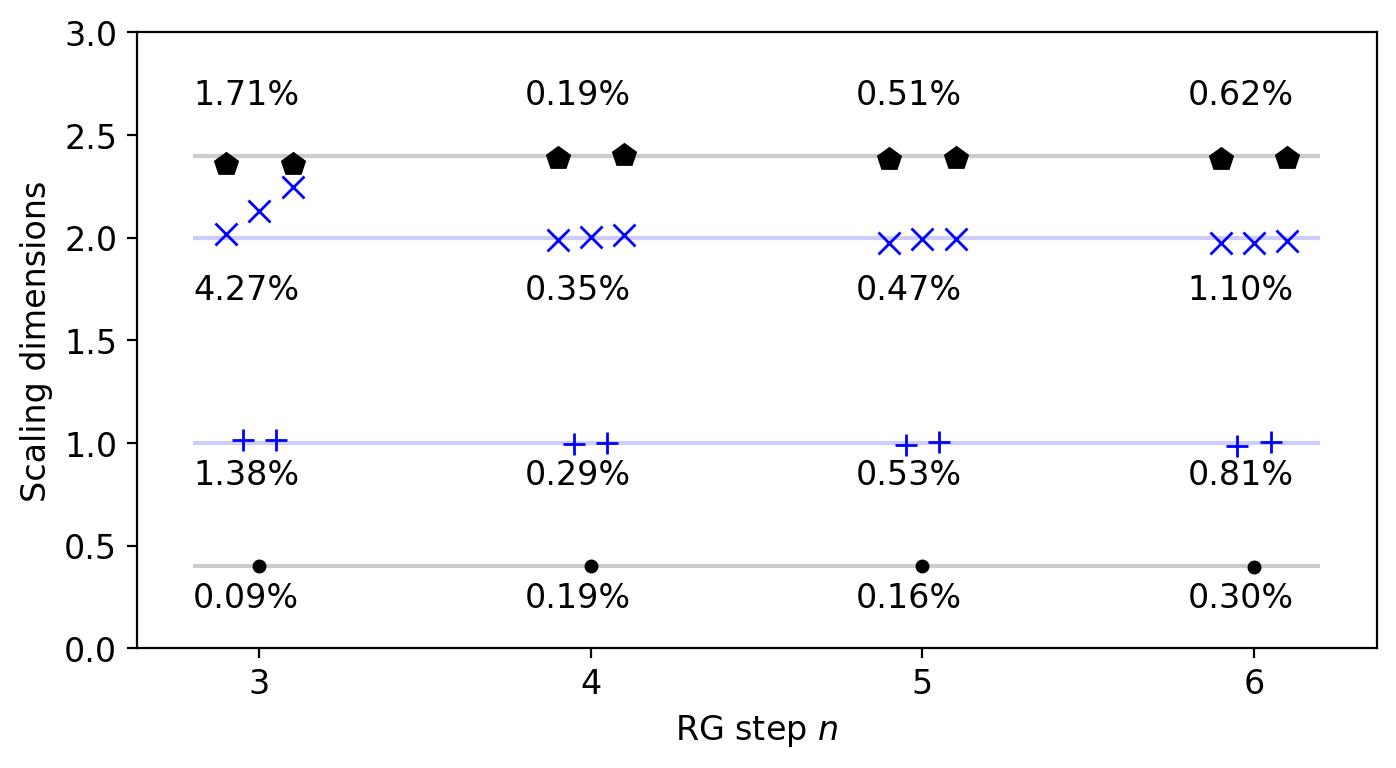}
}

\caption{\label{fig:zneg-scDim}
Estimates of the scaling dimensions at the negative-$z$ criticality using (a) the usual TRG and (b) the proposed method.
The tensor RG flow is generated at the estimated $z_c^{-}$ in Table~\ref{tab:estznc}.
The horizontal lines are exact values listed in Eq.~\eqref{eq:numres:scDznegExp}, while data points are the numerical estimates.
The percentages near the data points are the relative errors (we take the geometric average for the errors of degenerate scaling dimensions).
    }
\end{figure}

\section{Summary and discussion}
In this paper, we have demonstrated how to incorporate lattice-reflection, lattice-rotation, and $\mathcal{PT}$ symmetries in the TNRG.
We point out the importance of incorporating these symmetries in the TNRG for the phase transitions where they are spontaneously broken, as well as the fact that the physical meaning of keeping the tensor real-valued is the preservation of the $\mathcal{PT}$ symmetry of the model.
An EF-enhanced TNRG scheme is proposed where these symmetries are both preserved and imposed.
The scheme proposed in Sec.~\ref{sec:TNRG} works for any model whose tensor-network representation can be chosen to have the lattice and $\mathcal{PT}$ symmetries defined in Sec.~\ref{sec:TNRG:defsym}.
The lattice symmetries are expected to be present for models with interactions independent of the directions.
The $\mathcal{PT}$ symmetry is present for any statistical mechanical model with a real Boltzmann weight, but not necessarily positive.
This paper has made the first step of studying the hard-core lattice gas models using an advanced TNRG scheme enhanced with an EF process.
Combined with the hard-core break method~\cite{Akimenko:2023} for constructing the tensor-network representation of a hard-core lattice gas model, the proposed method could be used for studying such models with larger exclusion ranges.

We emphasize that there has been no EF-enhanced TNRG scheme that is suitable for studying a phase transition where some lattice symmetries are spontaneously broken.
While a simple scheme like the TRG naturally preserves lattice symmetries (if the machine precision can be ignored), all EF-enhanced schemes tend to break lattice symmetries as long as these symmetries are not explicitly incorporated.
The reason is that the EF generally involves an approximation of a non-tree-like tensor network, where the tensors in the tensor network are updated in series; the order of this update has to be chosen as an artifact of the algorithm and thus breaks the lattice symmetries.
Although this might not cause problems for a phase transition like that of the Ising model, it introduces unnecessary RG-relevant perturbations in the phase transitions where some lattice symmetries are broken spontaneously.
These RG-relevant perturbations due to the artifact of the algorithm get amplified with the coarse graining and can lead to incorrect RG flows\footnote{
        Our preliminary numerical experiments indicated that the Gilt-TNR~\cite{Hauru:2018} could not produce satisfactory RG flows near the positive-$z$ transition of the 1NN hard-square model;
        this observation motivated us to study the lattice-rotation symmetry in the TNRG.
}, posing difficulty in finding a critical fixed-point tensor.
As has been pointed out in Sec.~\ref{sec:TNRG:relationLTNR}, the ansatz proposed in the symmetric version of the loop-TNR~\cite{loop-TNR:2017} assumes a trivial bond matrix, which hinders its generality and might cripple its performance for models like hard-core lattice gas.

The scheme introduced in Appendix~\ref{app:loopopt} for determining the 3-leg tensor $\tilde{v}_{\Lambda}$ in the loop optimization in Eq.~\eqref{eq:scheme:loopopt} is still rudimentary.
It is adapted from the method proposed in Refs.~\cite{loop-TNR:2017,Evenbly:2018:fet}, which works best when various 3-leg tensors in the loop are independent of each other.
An ad hoc trick inspired by Ref.~\cite{Evenbly:code} is adopted to make the above method work when lattice symmetries are incorporated and all 3-leg tensors in the loop are the same one.
In our numerical calculations, we find the performance of the proposed updating rule for $\tilde{v}_{\Lambda}$ has potential for further improvement in terms of better convergence and reliability.
For example, the gradient descent, powered by automatic differentiation, can be combined with the proposed updating rule.
The updating rule prefers the low-rank solution but does not guarantee an increase of the fidelity.
The gradient descent ensures the increase of fidelity if possible but does not prefer a low-rank solution.
Therefore, a good way of combination is to first apply a minimal number of iterations of the proposed updating rule and then use the gradient descent to increase the fidelity.
An alternative possibility is adding the nuclear-norm regularization~\cite{Homma:2024} to the cost function of the gradient descent so that the low-rank solution becomes preferred.

\begin{acknowledgments}
The idea of developing a TNRG method with lattice-rotation symmetry was conceived while X.L. was pursuing his PhD at the University of Tokyo.
However, this paper was significantly shaped by Clement Delcamp's suggestion of understanding the spontaneous symmetry breaking in the hard-core lattice gas models using the tensor-network language.
We thank Slava Rychkov and Nikolay Ebel for fruitful discussions about rotation symmetry in tensor network calculations.
We thank Rajeev S. Erramilli and Sridip Pal for commenting about the entanglement-entropy area law for non-unitary conformal field theory.
We also thank Atsushi Ueda for explaining the optimization scheme used in the symmetric version of the loop-TNR.
X.L. is grateful for the beautiful and supportive research environment provided by  Institut des Hautes \'Etudes Scientifiques (IHES).
\end{acknowledgments}

\appendix

\section{Details of the loop optimization\label{app:loopopt}}
If the optimization procedure for the loop optimization in Eq.~\eqref{eq:scheme:loopopt} is carefully designed, the CDL tensors can be reliably filtered out without the additional entanglement filtering step in the loop-TNR proposed in Ref.~\cite{loop-TNR:2017}.

To understand this subtlety, first notice that the symmetric SVD splitting, as well as the usual SVD splitting, is exact for a CDL tensor even when the bond dimension of the diagonal leg is truncated from $\chi^2$ to $\chi$.
This means that the loop optimization in Eq.~\eqref{eq:scheme:loopopt} is also exact for a CDL tensor, with no further optimization needed.
Therefore, a simple gradient descent method will fail to simplify the CDL structure.
However, the optimization method developed in full environment truncation (FET)~\cite{Evenbly:2018:fet,Lyu:Kawashima:2025:reflsym} can reliably detect and filter out the CDL structure once certain subtleties are taken care of in its implementation.

In the framework of the FET, the two tensor networks on both sides of Eq.~\eqref{eq:scheme:loopopt} are treated as two ket vectors $\ket{\Psi}$ and $\ket{\Phi}$ by bending the physical legs of a tensor network towards the left:
\begin{align}
    \label{eq:scheme:psidef}
    \ket{\Psi} \equiv
    \includegraphics[scale=1.0, valign=c]{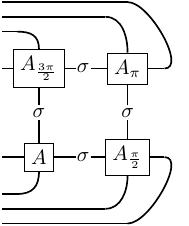},
\end{align}
and similarly for $\ket{\Phi}$.
The corresponding bra vectors $\bra{\Psi}$ and $\bra{\Phi}$ are mirror reflections of the ket diagrams.
The approximation in Eq.~\eqref{eq:scheme:loopopt} is quantified using the fidelity $F(\Psi,\Phi)$ between the two states, which is defined to be
\begin{align}
    \label{eq:scheme:Fdef}
    F(\Psi,\Phi)\equiv
    \frac{\braket{\Phi}{\Psi} \braket{\Psi}{\Phi}}{\braket{\Psi}{\Psi} \braket{\Phi}{\Phi}} \in [0,1].
\end{align}
The approximation error $\epsilon_{\text{loop}}^2$ can be defined to be $\epsilon_{\text{loop}}^2= 1- F$.
The optimal $\tilde{v}_{\Lambda}$ gives the maximal fidelity $F$.
The updating rule for $\tilde{v}_{\Lambda}$ is obtained by treating one copy of $\tilde{v}_{\Lambda}$ in the right-hand side of Eq.~\eqref{eq:scheme:loopopt} to be independent to the remaining copies; then, the fidelity can be rewritten as
\begin{subequations}
   \begin{align}
    \label{eq:scheme:F4vL}
    F(\Psi,\Phi)=
    \frac{\includegraphics[scale=1.0, valign=c]{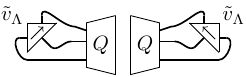}}{\includegraphics[scale=1.0, valign=c]{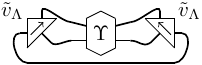}},
    \end{align}
    where the tensors $Q$ and $\Upsilon$ are defined as
    \begin{align}
    \label{eq:scheme:Qpsilondef}
    \includegraphics[scale=0.9, valign=c]{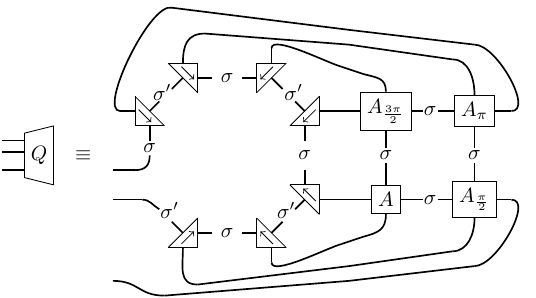},
    \end{align} 
    and
    \begin{align}
    \label{eq:scheme:Upsilondef}
    \includegraphics[scale=0.9, valign=c]{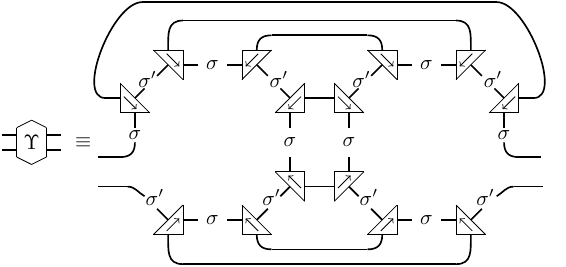}.
    \end{align} 
\end{subequations}
If one pretends that the copies of the tensor $\tilde{v}_{\Lambda}$ in the definition of $Q$ and $\Upsilon$ are different from the $\tilde{v}_{\Lambda}$ explicitly written in Eq.~\eqref{eq:scheme:F4vL}, the maximalization of $F$ in Eq.~\eqref{eq:scheme:F4vL} becomes a generalized eigenvalue problem~\cite{Lyu:Kawashima:2025:reflsym,geneig:2019};
the optimal $\tilde{v}_{\Lambda}$ is determined from tensors $Q, \Upsilon$:
\begin{align}
    \label{eq:scheme:vpropose}
    \includegraphics[scale=1.0, valign=c]{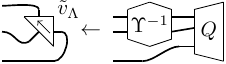}.
\end{align}
This updating rule is closely related to the standard variational MPS method used in the loop-TNR~\cite{loop-TNR:2017}, where $\tilde{v}_{\Lambda}$ is found by solving a linear system involving tensors $\Upsilon$ and $Q$.
Here we choose to use the inverse of $\Upsilon$ to solve this linear system.
Several subtleties need to be taken care of when using the updating rule in Eq.~\eqref{eq:scheme:vpropose}.

The existence of short-range entanglement located in the loop indicates the low-rank nature of the tensor $\Upsilon$ defined in Eq.~\eqref{eq:scheme:Upsilondef} when it is treated as a matrix horizontally~\cite{Hauru:2018}.
Therefore, the inverse of $\Upsilon$ in Eq.~\eqref{eq:scheme:vpropose} should always be implemented as the Moore-Penrose inverse in numerical calculations~\cite{Strang:2023,Lyu:Kawashima:2025:reflsym}.
Once the Moore-Penrose inverse is used, the updating rule in Eq.~\eqref{eq:scheme:vpropose} changes the CDL structure in a similar way to the graph-independent local truncation (Gilt)~\cite{Hauru:2018}.
The ``corner matrix'' in a CDL tensor gets multiplied by copies of itself after one iteration such that the singular value spectrum of the symmetric SVD splitting in the second symmetric TRG (see Eq.~\eqref{eq:scheme:oneRGloop}) has a lower entropy.
Usually, a few iterations using Eq.~\eqref{eq:scheme:vpropose} are enough to filter out a CDL tensor.
Therefore, in the numerical implementation, a minimal number of iterations\footnote{A minimum of 10 iterations works well in our numerical calculations.
} is set for the optimization of $\tilde{v}_{\Lambda}$.

In the derivation of the updating rule in Eq.~\eqref{eq:scheme:vpropose}, we pretend that the tensors $Q, \Upsilon$ do not depend on $\tilde{v}_{\Lambda}$.
If it were so, the updated $\tilde{v}_{\Lambda}$ guarantees the increase of the fidelity.
In reality, since the tensors $Q, \Upsilon$ change when $\tilde{v}_{\Lambda}$ is updated, the fidelity might not increase.
Here, we use a trick~\cite{Lyu:Kawashima:2025:reflsym,Evenbly:code} that makes sure that the fidelity does not decrease during the optimization process:
    \par
    (i) Use Eq.~\eqref{eq:scheme:vpropose} to propose a candidate $\tilde{v}_{\Lambda}'$.
    \par
    (ii) Form some convex combinations of $\tilde{v}_{\Lambda}'$ and the old $\tilde{v}_{\Lambda}$,
    \begin{align}
        \label{eq:scheme:vcombine}
        \tilde{v}_{\Lambda}^{\text{try}} =
        (1-p)\tilde{v}_{\Lambda}' + p\tilde{v}_{\Lambda},
    \end{align}
    where $p$ increases from 0 to 1.
    \par
    (iii) Once a convex combination leads to an increase of the fidelity, the tensor $\tilde{v}_{\Lambda}$ is updated to be this combination.


In Sec.~\ref{sec:numres}, we use two different ways to implement the Moore-Penrose inverse of $\Upsilon$ for the two transitions of the model.
The Moore-Penrose inverse of a symmetric matrix $M$ can be obtained from its eigendecomposition (ED). 
Suppose the ED of the matrix is $M=O\Lambda O^T$, where $O$ is orthogonal and $\Lambda =\diag{(\lambda_1, \lambda_2,\ldots)}$ is diagonal with $|\lambda_1| \geq |\lambda_2| \geq \ldots$.
Our first way to implement the Moore-Penrose inverse $M^{+}$ is using a ``hard inverse'' of the diagonal matrix $\Lambda$ in the ED, $M^+ = O\Lambda^{+} O^T$, with
\begin{align}
    \label{eq:scheme:pinvhard}
    \Lambda^{+} = \diag{(\lambda_1^{-1}, \lambda_2^{-1},\dots,\lambda_k^{-1},0,0,\dots)},
\end{align}
where the first $k$ eigenvalues are inverted and all the remaining ones are set to be zero.
In the numerical calculation of the positive-activity transition in Sec.~\ref{sec:numres:posz}, we find a good choice of $k$ is $k=4\times \floor{\frac{\chi_1}{3}} 
\times \floor{\frac{\chi_2}{3}}$, where $\chi_1$ and $\chi_2$ are bond dimensions of the two legs on the same side of $\Upsilon$ in Eq.~\eqref{eq:scheme:Upsilondef}.
This choice is convenient since it introduces no hyperparameter for the loop optimization.

However, the ``hard-inverse" construction in Eq.~\eqref{eq:scheme:pinvhard} does not work well for negative-activity transition, where the loop optimization often gets stuck with the TRG initialization of $\tilde{v}_{\Lambda}$.
Inspired by Ref.~\cite{Hauru:2018}, we propose a second way to implement the Moore-Penrose inverse $M^{+'}$ using a ``regulated inverse'' of the diagonal matrix $\Lambda$ in the ED, $M^{+'} = O\Lambda^{+'} O^T$, with
\begin{align}
    \label{eq:scheme:pinvsoft}
    (\Lambda^{+'})_k = \frac{1}{\lambda_k + \epsilon_{\text{inv}}},
\end{align}
where $\epsilon_{\text{inv}}$ is a small positive number as the regulator of the inverse.
In the numerical calculation of the negative-activity transition in Sec.~\ref{sec:numres:negz}, we observe that the effectiveness of the updating rule in Eq.~\eqref{eq:scheme:vpropose} is very sensitive to the choice of $\epsilon_{\text{inv}}$.
In Table~\ref{tab:znceps}, we summarize the good choices of $\epsilon_{\text{inv}}$ for various bond dimensions.
    A good choice of $\epsilon_{\text{inv}}$ means that the loop optimization is effective so that the error of the TRG is reduced compared with the usual TRG without the EF and the growth of the error with the RG step becomes slower (see Fig.~\ref{fig:zneg-RGerrs}).

\begin{table}[tb]
\caption{\label{tab:znceps}%
The parameter $\epsilon_{\text{inv}}$ used in the Moore-Penrose ``regulated inverse'' for negative-activity transition in Sec.~\ref{sec:numres:negz}.
}
\begin{ruledtabular}
\begin{tabular}{ccccc}
$\chi$ & 10 &  12 & 16 & 20\\
$\epsilon_{\text{inv}}$ & $1\times 10^{-8}$ & $1\times 10^{-10}$ & $5\times 10^{-11}$ & $6\times 10^{-12}$
\end{tabular}
\end{ruledtabular}
\end{table}

\section{From two sublattices to one uniform lattice\label{app:buildTM}}
In the proposed TNRG in Sec.~\ref{sec:TNRG}, the coarse-grained tensor network representation of the partition function takes the form in Eq.~\eqref{eq:scheme:initialTN}, where two types of tensors $A$ and $B$ occupy two sublattices, respectively.
One simple way to construct the transfer matrix of the partition function in Eq.~\eqref{eq:scheme:initialTN} is using a $2 \times 2$ block of tensors consisting of two copies of $A$, two copies of $B$, and eight copies of the bond matrix $\sigma$.
This way is unavoidable when tensors $A$ and $B$ are unrelated to each other. 
However, once the tensors $A$ and $B$ have the strong form of the lattice symmetries in Eq.~\eqref{eq:scheme:strongsym}, one can choose to make the tensor network in Eq.~\eqref{eq:scheme:initialTN} have just one type of tensor $A_{\text{uni}}$ occupying one lattice uniformly.
Afterwards, copies of a single tensor $A_{\text{uni}}$, instead of a much larger $2 \times 2$ block, can be used to build the transfer matrix.

To map from two sublattices to one, use Eq.~\eqref{eq:scheme:ABstrong} in Eq.~\eqref{eq:scheme:initialTN} to replace all $B$ with $A$;
the resultant tensor network is
\begin{align}
    \label{eq:num:uniTN}
    \includegraphics[scale=0.9, valign=c]{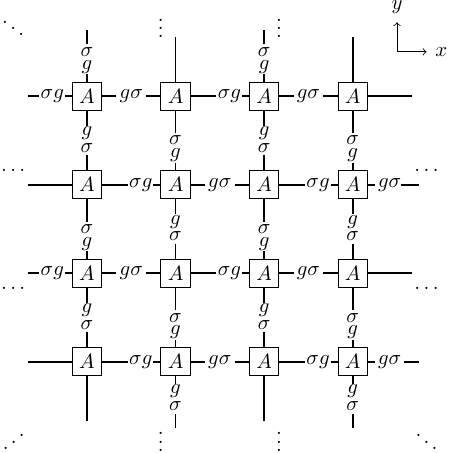},
\end{align}
where two matrices $g$ and $\sigma$ sit on each bond.
These two matrices commute with each other since they are both diagonal.
Define their product to be $\tilde{\sigma} \equiv g \sigma = \sigma g$ and let the tensor $A$ absorb $\tilde{\sigma}$ to be a new tensor $A_{\text{uni}}$:
\begin{align}
    \label{eq:num:A2C}
    \includegraphics[scale=1.0, valign=c]{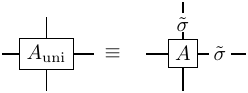}.
\end{align}
The tensor-network representation of the partition function in Eq.~\eqref{eq:scheme:initialTN} becomes a uniform tensor network with just one type of tensor $A_{\text{uni}}$, and the transfer matrix can be built in the usual way using copies of the tensor $A_{\text{uni}}$.


\bibliography{reference}

\end{document}